\documentclass[envcountsame]{llncs}
\usepackage{hyperref}
\usepackage{amsmath, amssymb}
\usepackage{epsfig,psfrag}
\usepackage{enumerate}
\usepackage{multirow}
\newcommand{\bs}{\boldsymbol}
\DeclareMathOperator*{\argmax}{arg\,max}

\begin{document}
\title{Rank-1 Bimatrix Games: A Homeomorphism and \\a Polynomial Time Algorithm}
\author{Bharat Adsul \and Jugal Garg \and Ruta Mehta \and Milind Sohoni\\ Indian Institute of Technology, Bombay\\
adsul, jugal, ruta, sohoni@cse.iitb.ac.in}
\maketitle
\begin{abstract}
Given a rank-1 bimatrix game $(A,B)$, {\it i.e.,} where $rank(A+B)=1$, we construct a suitable linear subspace of the rank-1
game space and show that this subspace is homeomorphic to its Nash equilibrium correspondence. Using this homeomorphism, we
give the first polynomial time algorithm for computing an exact Nash equilibrium of a rank-1 bimatrix game. This settles an
open question posed in \cite{kan,the}. In addition, we give a novel algorithm to enumerate all the Nash equilibria of a
rank-1 game and show that a similar technique may also be applied for finding a Nash equilibrium of any bimatrix game. This
technique also proves the existence, oddness and the index theorem of Nash equilibria in a bimatrix game. Further, we extend
the rank-1 homeomorphism result to a fixed rank game space, and give a fixed point formulation on $[0,1]^k$ for solving a
rank-$k$ game. The homeomorphism and the fixed point formulation are piece-wise linear and considerably simpler than the
classical constructions.
\end{abstract}

\section{Introduction}
Non-cooperative game theory is a model to understand strategic interaction of selfish agents in a given organization. 
In a finite game, there are finitely many agents, each having finitely many strategies. For finite games, Nash \cite{nash}
proved that there exists a steady state where no player benefits by a unilateral deviation. Such a steady state is called a
Nash equilibrium of the game. 

Finite games with two agents are also called bimatrix games since they may be represented by two payoff matrices
$(A,B)$, one for each agent. The problem of computing a Nash equilibrium of a bimatrix game is said to be one of the most
important concrete open questions on the boundary of $\mathcal P$ \cite{pap}. The classical Lemke-Howson (LH) algorithm
\cite{lem} finds a Nash equilibrium of a bimatrix game. However, Savani and von Stengel \cite{sav} showed that it is not a
polynomial time algorithm by constructing an example, for which the LH algorithm takes an exponential number of steps. Chen
and Deng \cite{chen} showed that this problem is $\mathcal{PPAD}$-complete, a complexity class introduced by Papadimitriou
\cite{pap1}. They (together with Teng) \cite{chen1} also showed that the computation of even a
$\frac{1}{n^{\Theta(1)}}$-approximate Nash equilibrium remains $\mathcal{PPAD}$-complete. These results suggest that a
polynomial time algorithm is unlikely.

There are some results for special cases of the bimatrix games. Lipton et al. \cite{lip} considered games where both payoff
matrices are of fixed rank $k$ and for these games, they gave a polynomial time algorithm for finding a Nash equilibrium.
However, the expressive power of this restricted class of games is limited in the sense that most zero-sum games are not
contained in this class. Kannan and Theobald \cite{kan} defined a hierarchy of bimatrix games using the rank of $(A+B)$ and
gave a polynomial time algorithm to compute an approximate Nash equilibrium for games of a fixed rank $k$. The set of
rank-$k$ games consists of all the bimatrix games with rank at most $k$. Clearly, rank-$0$ games are the same as zero-sum
games and it is known that the set of Nash equilibria of a zero-sum game is a polyhedral set (hence, connected) and it may be
computed in polynomial time by solving a linear program (LP). Moreover, the problem of finding a Nash equilibrium of 
zero-sum games and solving linear programs are equivalent \cite{dan}. 

The set of rank-$1$ games is the smallest extension of zero-sum games in the hierarchy, which strictly generalizes zero-sum
games. For any given constant $c$, Kannan and Theobald \cite{kan} also construct a rank-$1$ game, for which the number of
connected components of Nash equilibria is larger than $c$. This shows that the expressive power of rank-$1$ games is 
larger than the zero-sum games. 
Rank-$1$ games may also arise in practical situations, in particular the {\em multiplicative games} between firms and workers
in \cite{bulow} are rank-$1$ games. 
A polynomial time algorithm to compute an exact Nash equilibrium for rank-$1$ games is an important open problem
\cite{kan,the}. Kontogiannis and Spirakis \cite{ks} defined the notion of mutual (quasi-) concavity of a bimatrix game and
for mutual (quasi-) concave games, they provide a polynomial time (FPTAS) computation of a Nash equilibrium, however their
classification and the games of fixed rank are incomparable.

Shapley's index theory \cite{shap} assigns a sign (also called an index) to a Nash equilibrium of a bimatrix game and shows
that the indices of the two endpoints of a Lemke-Howson path have opposite signs. The signs of the endpoints of LH paths
provide a direction and in turn a ``parity argument" that puts the Nash equilibrium problem of a bimatrix game in
$\mathcal{PPAD}$ \cite{pap1,stengel}.
The set of bimatrix games $\Omega$, where the number of strategies of the first and second players are $m$ and $n$
respectively, forms a $\mathbb R^{2mn}$ Euclidean space, {\it i.e.,} $\Omega=\{(A,B)\in \mathbb R^{mn}\times \mathbb
R^{mn}\}$. Kohlberg and Mertens \cite{koh} showed that $\Omega$ is homeomorphic to its Nash equilibrium
correspondence\footnote{The actual result is for $N$ player game space.} $E_{\Omega}=\{(A,B,x,y)\in \mathbb R^{2mn+m+n}\ |\
(x,y) \mbox{ is a Nash equilibrium of }(A,B)\}$. 

This structural result has been used extensively to understand the index,
degree and the stability of a Nash equilibrium of a bimatrix game \cite{g2,koh}. Moreover, the homeomorphism result also
validates the homotopy methods devised to compute a Nash equilibrium \cite{gov,sur}. 
The structural result has been extended for more general game spaces \cite{ak}, however, to the best of our knowledge, no such
result is known for special subspaces of the bimatrix game space. Such a result may pave a way to device a better algorithm
for the Nash equilibrium computation or to prove the hardness of computing a Nash equilibrium, for the games in the subspace.
\\ \\
\noindent {\bf Our contributions.} 
For a given rank-$1$ game $(A,B)\in\mathbb R^{mn}\times \mathbb R^{mn}$, the matrix $(A+B)$ may be written as
$\alpha\cdot\beta^T$, where $\alpha\in\mathbb R^m$ and $\beta\in\mathbb R^n$. Motivated by this fact, in Section
\ref{games_polytopes}, we define an $m$-dimensional subspace $\Gamma=\{(A,C+\alpha.\beta^T)\ |\ \alpha\in\mathbb R^m\}$ of
$\Omega$, where $A\in\mathbb R^{mn},C\in\mathbb R^{mn}$ and $\beta\in\mathbb R^n$ are fixed and analyze the structure of its
Nash equilibrium correspondence $E_\Gamma$. For a given bimatrix game $(A,B)$, the best response polytopes $P$ and $Q$ may be
defined using the payoff matrices $A$ and $B$ respectively \cite{agt} (also in Section \ref{prel}). There is a notion of
fully-labeled points of $P\times Q$, which capture all the Nash equilibria of the game. Note that the polytope $P$ is same
for all the games in $\Gamma$ since the payoff matrix of the first player is fixed to $A$. However the payoff matrix
of the second player varies with $\alpha$, hence $Q$ is different for every game. We define a new polytope $Q'$ in Section
\ref{games_polytopes}, which encompasses $Q$ for all the games in $\Gamma$. We show that the set of fully-labeled points
of $P\times Q'$, say $\mathcal N$, captures all the Nash equilibria of all the games in $\Gamma$ and in turn captures
$E_\Gamma$.  

Surprisingly, $\mathcal N$ turns out to be a set of cycles and a single path on the $1$-skeleton of $P\times Q'$. We refer to the path
in $\mathcal N$ as the fully-labeled path and show that it contains at least one Nash equilibrium of every game in $\Gamma$.
The structure of $\mathcal N$ also proves the existence and the oddness of the number of Nash equilibria in a non-degenerate
bimatrix game. Moreover, an edge of $\mathcal N$ may be efficiently oriented, and using this orientation, we determine the
index of every Nash equilibria for a bimatrix game. 
Further, in Section \ref{hom} we show that if $\Gamma$ contains only rank-1 games ({\em i.e.}, $C=-A$) then $\mathcal N$ does
not contain any cycle and the fully-labeled path exhibits a strict monotonicity. Using this monotonic nature, we establish
homeomorphism maps between $\Gamma$ and $E_\Gamma$. This is the first structural result for a subspace of the bimatrix game
space. The homeomorphism maps that we derive are very different than the ones given by Kohlberg and Mertens for the bimatrix
game space \cite{koh}, and require a structural understanding of $E_\Gamma$. 

Using the above facts on the structure of $\mathcal N$, in Section \ref{algo} we present two algorithms. For a given rank-$1$
game $(A,-A+\gamma.\beta^T)$, we consider the subspace $\Gamma=\{(A,-A+\alpha.\beta^T)\ |\ \alpha\in\mathbb R^m\}$. Note that
$\Gamma$ contains the given game and the corresponding set $\mathcal N$ is a path which captures all the Nash equilibria of
the game. The first algorithm ({\em BinSearch}) finds a Nash equilibrium of a rank-$1$ game in polynomial time by applying
binary search on the fully-labeled path using the monotonic nature of the path. To the best of our knowledge, this is the
first polynomial time algorithm to find an exact Nash equilibrium of a rank-$1$ game. 

The second algorithm ({\em Enumeration}) enumerates all the Nash equilibria of a rank-$1$ game. Using the fact that $\mathcal
N$ contains only the fully-labeled path, the {\em Enumeration} algorithm traces this path and locates all the Nash equilibria
of the game.  For an arbitrary bimatrix game, we may define a suitable $\Gamma$ containing the game. Since the fully-labeled
path of the corresponding $\mathcal N$ covers at least one Nash equilibrium of all the games in $\Gamma$, the {\em
Enumeration} algorithm locates at least one Nash equilibrium of the given bimatrix game.
Theobald \cite{the} also gave an algorithm to enumerate all the Nash equilibria of a rank-$1$ game, however it may not be
generalized to find a Nash equilibrium of any bimatrix game. Moreover, our algorithm is much simpler and a detailed
comparison is given in Section \ref{enum}. There, we also compare our algorithm with the Lemke-Howson algorithm, which
follows a path of almost\footnote{For a fixed label $1\leq r\leq m+n$, all the labels except $r$ should be present.}
fully-labeled points \cite{agt}. 

For a given rank-$k$ game $(A,B)$, the matrix $(A+B)$ may be written as $\sum_{l=1}^{k} \gamma^l.\beta^{l^T}$, where $\forall
l, \gamma^{l}\in\mathbb R^{m}$ and $\beta^{l}\in\mathbb R^{n}$. We define a $km$-dimensional affine subspace
$\Gamma^k=\{(A,-A+\sum_{l=1}^{k} \alpha^l.\beta^{l^T})\ |\ \alpha^l \in\mathbb R^{m}, \forall l\}$ of $\Omega$. In Section
\ref{fixedRank}, we establish a homeomorphism between $\Gamma^k$ and its Nash equilibrium correspondence $E_{\Gamma^k}$
using techniques similar to the rank-$1$ homeomorphism. Further, to find a Nash equilibrium of a rank-$k$ game we give a
piece-wise linear polynomial-time computable fixed point formulation on $[0, 1]^k$ using the homeomorphism result and
discuss the possibility of a polynomial time algorithm. 

\section{Games and Nash Equilibrium}
\subsection{Preliminaries}\label{prel}
{\bf Notations.} For a matrix $A=[a_{ij}] \in \mathbb R^{mn}$ of dimension $m\times n$, let $A_i$ be the $i^{th}$ row and
$A^j$ be the $j^{th}$ column of the matrix. Let $0_{l\times k}$ and $1_{l \times k}$ be the matrices of dimension $l\times k$
with all $0$s and all $1$s respectively. For a vector $\alpha \in \mathbb R^{m}$, let $\alpha_i$ be its $i^{th}$ coordinate.
Vectors are considered as column vectors. 

For a finite two players game, let the strategy sets of the first and the second player be $S_1=\{1,\dots,m\}$ and
$S_2=\{1,\dots,n\}$ respectively. The payoff function of such a game may be represented by the two payoff matrices $(A,B)
\in \mathbb R^{mn} \times \mathbb R^{mn}$, each of dimension $m \times n$. If the played strategy profile is $(i,j)\in
S_1\times S_2$, then the payoffs of the first and second players are $a_{ij}$ and $b_{ij}$ respectively. Note that the rows
of these matrices correspond to the strategies of the first player and the columns to that of the second player, hence the
first player is also referred to as the {\em row-player} and second player as the {\em column-player}.

A {\em mixed strategy} is a probability distribution over the available set of strategies. The set of mixed strategies for
the row-player is $\Delta_1=\{(x_1,\dots,x_m)\ |\ x_i\geq 0,\ \forall i \in S_1,\ \sum_{i=1}^m x_i=1 \}$ and for the
column-player, it is $\Delta_2=\{(y_1,\dots,y_n)\ |\ y_j\geq 0,\ \forall j \in S_2,\ \sum_{j=1}^n y_j=1\}$. The strategies in
$S_1$ and $S_2$ are called {\em pure strategies}. If the strategy profile $(x,y) \in \Delta_1 \times \Delta_2$ is played, then
the payoffs of the row-player and column-player are $x^TAy$ and $x^TBy$ respectively.

A strategy profile is said to be a Nash equilibrium strategy profile (NESP) if no player achieves a better payoff by a
unilateral deviation \cite{nash}. Formally, $(x,y) \in \Delta_1\times \Delta_2$ is a NESP iff $\forall x' \in \Delta_1,\ 
x^TAy \geq x'^TAy$ and $\forall y' \in \Delta_2,\  x^TBy \geq x^TBy'$. 
These conditions may also be equivalently stated as follows.
\vspace{-0.1cm}
\begin{eqnarray}\label{eq1_equiv}
\begin{array}{ll}
\forall i \in S_1,\hspace{.06in} x_i>0 \ \ \Rightarrow\ \ & A_iy =\max_{k \in S_1} A_ky \\
\forall j \in S_2,\hspace{.06in} y_j>0 \ \ \Rightarrow & x^TB^j = \max_{k \in S_2} x^TB^k
\end{array}
\end{eqnarray}

From (\ref{eq1_equiv}), it is clear that at a Nash equilibrium, a player plays a pure strategy with non-zero probability only
if it gives the maximum payoff with respect to (w.r.t.) the opponent's strategy. Such strategies are called the {\em best
response} strategies (w.r.t. the opponent's strategy). The polytope $P$ in (\ref{eq2}) is closely related to the best
response strategies of the row-player for any given strategy of the column-player \cite{agt} and it is called the {\em best
response} polytope of the row-player. Similarly, the polytope $Q$ is called the best response polytope of the column-player.
In the following expression, $x$ and $y$ are vector variables, and $\pi_1$ and $\pi_2$ are scalar variables.
\begin{eqnarray}\label{eq2}
\begin{array}{llclcll}
P=\{&(y,\pi_1)\in \mathbb R^{n+1}\ \hspace{3pt}|& A_iy-\pi_1\leq 0,& \forall i \in S_1;& y_j\geq 0,& \forall j \in S_2;& \hspace{.06in}\sum_{j=1}^n y_j=1\}\\
Q=\{&(x,\pi_2)\in \mathbb R^{m+1}\ |& x_i\geq 0,& \forall i \in S_1;&\hspace{.06in} x^TB^j-\pi_2\leq 0,& \forall j \in S_2;& \hspace{.06in}\sum_{i=1}^m x_i=1\}
\end{array}
\end{eqnarray}

Note that for any $y' \in \Delta_2$, a unique $(y',\pi'_1)$ may be obtained on the boundary of $P$, where
$\pi'_1=\max_{i\in S_1} A_iy'$.  Clearly, the pure strategy $i \in S_1$ is in the best response against $y'$ only if
$A_iy'-\pi'_1=0$, hence indices in $S_1$ corresponding to the tight inequalities at $(y',\pi'_1)$ are in the best response.
Note that, in both the polytopes the first
set of inequalities correspond to the row-player, and the second set correspond to the column player. Since $|S_1|=m$ and
$|S_2|=n$, let the inequalities be numbered from $1$ to $m$, and $m+1$ to $m+n$ in both the polytopes. Let the {\em label}
$L(v)$ of a point $v$ in the polytope be the set of indices of the tight inequalities at $v$. If a pair $(v,w) \in P \times
Q$ is such that $L(v)\cup L(w)=\{1,\dots,m+n\}$, then it is called a {\em fully-labeled pair}. 

\begin{lemma}\label{ne_le}
A strategy profile $(x,y)$ is a NESP of the game $(A,B)$ iff $((y,\pi_1),(x,\pi_2))\in P\times Q$ is a fully-labeled pair, 
for some $\pi_1$ and $\pi_2$ \cite{agt}. 
\end{lemma}

A game is called non-degenerate if both the polytopes are non-degenerate. Note that for a non-degenerate game, $|L(v)|\leq n$
and $|L(w)|\leq m$, $\forall (v,w) \in P\times Q$, and the equality holds iff $v$ and $w$ are the vertices of $P$ and $Q$
respectively. Therefore, a fully-labeled pair of a non-degenerate game has to be a vertex-pair. However, for a degenerate
game, there may be a fully-labeled pair $(v,w)$, which is not a vertex-pair. In that case, if $v$ is on a face of $P$, then
every point $v'$ of this face makes a fully-labeled pair with $w$ since $L(v)\subseteq L(v')$.  Similarly, if $w$ is on a
face of $Q$, then every point $w'$ of this face makes a fully-labeled pair with $v$. 

Let $E=\{(A,B,x,y)\in \mathbb R^{mn}\times\mathbb R^{mn}\times \Delta_1\times \Delta_2\ |\ (x,y) \mbox{ is a NESP of the game
}(A,B)\}$ be the Nash equilibrium correspondence of the bimatrix game space $\mathbb R^{2mn}$ ({\it i.e.,} $\mathbb
R^{mn}\times \mathbb R^{mn}$). Kohlberg and Mertens \cite{koh} proved that $E$ is homeomorphic to the bimatrix game space
$\mathbb R^{2mn}$. No such structural result is known for a subspace of the bimatrix game space $\mathbb R^{2mn}$. 
To extend such a result for a subspace, in the next section, we define an $m$-dimensional affine subspace of $\mathbb R^{2mn}$
and analyze the structure of it's Nash equilibrium correspondence.

\subsection{Game Space and the Nash Equilibrium Correspondence}\label{games_polytopes}
Let $\Gamma=\{(A,C+\alpha\cdot\beta^T)\ |\ \alpha \in \mathbb R^m\}$ be a game space, where $A\in\mathbb R^{mn}$ and
$C\in\mathbb R^{mn}$ are $m \times n$ dimensional non-zero matrices, and $\beta\in \mathbb R^n$ is an $n$-dimensional
non-zero vector.  Note that for a game $(A,B)\in\Gamma$, there exists a unique $\alpha \in \mathbb R^m$, such that
$B=C+\alpha\cdot\beta^T$.  Therefore, $\Gamma$ may be parametrized by $\alpha$, and let $G(\alpha)$ be the game
$(A,C+\alpha\cdot\beta^T) \in\Gamma$.  Clearly, $\Gamma$ forms an $m$-dimensional affine subspace of the bimatrix game space
$\mathbb R^{2mn}$. Let $E_\Gamma=\{(\alpha,x,y)\in \mathbb R^m\times \Delta_1\times \Delta_2\ |\ (x,y) \mbox{ is a NESP of
the game }G(\alpha) \in \Gamma\}$ be the Nash equilibrium correspondence of $\Gamma$. 
We wish to investigate: {\em Is $E_\Gamma$  homeomorphic to the game space $\Gamma$ ($\equiv \mathbb R^m$)?} 

For a game $G(\alpha) \in \Gamma$, let the best response polytopes of row-player and column-player be $P(\alpha)$ and
$Q(\alpha)$ respectively. Since the row-player's matrix is fixed to $A$, hence $P(\alpha)$ is the same for all $\alpha$ and we
denote it by $P$. However, $Q(\alpha)$ varies with $\alpha$. We define a new polytope $Q'$ in (\ref{eq3}), which encompasses
$Q(\alpha)$, for all $G(\alpha) \in \Gamma$.
\begin{eqnarray}\label{eq3}
Q'=\{(x,\lambda,\pi_2)\in \mathbb R^{m+2}\ |\ x_i\geq 0,\ \forall i \in S_1;\ x^TC^j+\beta_j\lambda-\pi_2\leq 0,\ \forall j
\in S_2;\ \sum_{i=1}^m x_i=1\} \end{eqnarray}

Note that the inequalities of $Q'$ may also be numbered from $1$ to $m+n$ in a similar fashion as in $Q$.
For a game $G(\alpha)$, the polytope $Q(\alpha)$ may be obtained by replacing $\lambda$ by $\sum_{i=1}^m \alpha_i x_i$ in
$Q'$. In other words, $Q(\alpha)$ is the projection of $Q' \cap \{(x,\lambda,\pi_2)\ |\ \sum_{i=1}^m \alpha_i x_i - \lambda
= 0\}$ on the $(x,\pi_2)$-space. Let $\mathcal N=\{(v,w)\in P\times Q'\ |\ L(v)\cup L(w)=\{1,\dots,m+n\}\}$ be the set of
fully-labeled pairs in $P\times Q'$. The following lemma relates $E_{\Gamma}$ and $\mathcal N$.

\begin{lemma}\label{le1}
\begin{enumerate}
\item If $((y,\pi_1),(x,\lambda,\pi_2))\in \mathcal N$, then there is an $\alpha \in \mathbb R^m$ such that $(\alpha,x,y)\in E_\Gamma$.
\item For every $(\alpha,x,y)\in E_\Gamma$, there exist unique $\pi_1$, $\pi_2$ and $\lambda$ in $\mathbb R$, s.t. $((y,\pi_1),(x,\lambda,\pi_2))\in \mathcal N$.
\end{enumerate}
\end{lemma}
\begin{proof} For the first part, suppose $(v,w)$ is a fully-labeled pair with $v=(y,\pi_1)$ and $w=(x,\lambda,\pi_2)$. Let 
$\alpha \in \mathbb R^m$ be such that $\sum_{i=1}^m\alpha_i x_i-\lambda=0$, then clearly $(v,(x,\pi_2))\in P(\alpha)\times
Q(\alpha)$ is a fully-labeled pair. Therefore, $(\alpha,x,y) \in E_\Gamma$.

For the second part, let $(\alpha,x,y) \in E_\Gamma$, then for $\pi_1=x^TAy$ and $\pi_2=x^T(C+\alpha\cdot\beta^T)y$, we get a
fully labeled pair $((y,\pi_1),(x,\pi_2))\in P(\alpha)\times Q(\alpha)$. Hence, for $\lambda=\sum_{i=1}^m \alpha_i x_i$,
the point $((y,\pi_1),(x,\lambda,\pi_2))$ is in $\mathcal N$.
\qed
\end{proof}

From Lemma \ref{le1}, it is clear that there is a continuous surjective map from $E_{\Gamma}$ to $\mathcal N$. We further
strengthen the connection in the following lemma.

\begin{lemma}\label{le_con}
$E_\Gamma$ is connected iff $\mathcal N$ is a single connected component.
\end{lemma}
\begin{proof}
($\Rightarrow$) Lemma \ref{le1} shows that for a point $(\alpha,x,y) \in E_\Gamma$, we may construct a unique point
$((y,\ x^TAy),(x,\ x^T\alpha,\ x^T(C+\alpha\cdot\beta^T)y)) \in \mathcal N$. This gives a continuous surjective function
$f:E_\Gamma \rightarrow \mathcal N$. Therefore, if $E_\Gamma$ is connected then $\mathcal N$ is connected as well.

($\Leftarrow$) For a $(v,w)\in \mathcal N$, where $w=(x,\lambda,\pi_2)$, all the points in $f^{-1}(v,w)$ satisfy
$\sum_{i=1}^m x_i \alpha_i=\lambda$, hence $f^{-1}(v,w)$ is homeomorphic to $\mathbb R^{m-1}$. Since $\mathcal N$ is connected,
$f$ is continuous and the fact that the fibers $f^{-1}(v,w),\ \forall (v,w)\in \mathcal N$ are connected imply that
$E_\Gamma$ is connected.
\qed
\end{proof}

Lemma \ref{le1} and \ref{le_con} imply that $E_\Gamma$ and $\mathcal N$ are closely related.
Henceforth, we assume that the polytopes $P$ and $Q'$ are
non-degenerate. Recall that when the best response polytopes ($P$ and $Q$) of a game are non-degenerate, all the
fully-labeled pairs are vertex pairs. However $Q'$ has one more variable $\lambda$ than $Q$, which gives one extra degree of
freedom to form the fully-labeled pairs. We show that the structure of $\mathcal N$ is very simple by proving the following
proposition. 

\begin{proposition}\label{pr1}
The set of fully-labeled points $\mathcal N$ admits the following decomposition into mutually disjoint connected components. 
$\mathcal N=\mathcal P \cup \mathcal C_1 \cup \dots\cup \mathcal C_k,\ k\geq 0$, where $\mathcal P$ and $\mathcal C_i$s
respectively form a path and cycles on $1$-skeleton of $P\times Q'$. 
\end{proposition}

In order to prove Proposition \ref{pr1}, first we identify the points in $P$ and $Q'$ separately, which participate in
the fully-labeled pairs and then relate them. 
For a $v\in P$, let $\mathcal E_v=\{w' \in Q'\ |\ (v,w')\in \mathcal N\}$, and similarly for a $w \in Q'$, let $\mathcal
E_w=\{v' \in P\ |\ (v',w) \in \mathcal N\}$. 
Let $\mathcal N^P=\{v \in P \ |\ \mathcal E_v\neq \emptyset\}$ and $\mathcal N^{Q'}=\{w \in Q' \ |\ \mathcal E_w\neq
\emptyset\}$.

For neighboring vertices $u$ and $v$ in either polytopes, let $\overline{u,v}$ be the edge between $u$ and $v$. Recall
that $P$ and $Q'$ are non-degenerate, therefore $\forall v \in P,\ |L(v)|\leq n$ and $\forall w \in Q',\ |L(w)|\leq m+1$.
Using this fact, it is easy to deduce the following observations for points in $P$.
Similar results hold for the points in $Q'$. 

\begin{itemize}
\item[$O_1$.] If $(v,w) \in \mathcal N$, then both $v$ and $w$ lie on either $0$ or $1$-dimensional faces of $P$ and $Q'$
respectively, and at least one of them is a $0$-dimensional face, {\em i.e.}, a vertex. 
\item[$O_2$.] If $(v,w)\in \mathcal N$ and both $v$ and $w$ are vertices, then $|L(v)\cap L(w)|=1$, and the element in the
intersection is called the {\em duplicate label} of the pair $(v,w)$.  
\item[$O_3$.] If $v\in P$ is not a vertex then $\mathcal E_v$ is either empty or it equals exactly one vertex of $Q'$. 
\item[$O_4$.] If $v\in P$ is a vertex, then $\mathcal E_v$ is either empty or an edge of $Q'$.
\item[$O_5$.] Let $v\in P$ be a vertex and $\mathcal E_v$ be an edge of $Q'$. If $w\in\mathcal E_v$ is a vertex, then $(v,w)$
has a duplicate label (see $O_2$). Let the duplicate
label be $i$, then there exists a unique vertex $v' \in P$ adjacent to $v$ such that $\overline{v,v'} \in \mathcal N^P$, where $v'$ is
obtained by relaxing the inequality $i$ at $v$. This also implies that $\mathcal E_w=\overline{v,v'}$ and $\mathcal E_v\cap
\mathcal E_{v'}=w$.
\end{itemize}

The above observations, brings out the structure of $\mathcal N$ significantly. 
Every point in $\mathcal N$ is a pair $(v,w)$ where $v \in P$ and $w\in Q'$. From $O_1$,
one of them is a vertex (say $v$), and the other is on the corresponding edge ($w\in \mathcal E_v$). Hence $\mathcal N$
contains only $0$ and $1$-dimensional faces of $P\times Q'$. Clearly, an edge of $\mathcal N$ is of type $(v,\mathcal E_v)$
or $(\mathcal E_w,w)$, where $v$ and $w$ are the vertices of $P$ and $Q'$ respectively. 

Note that a vertex $(v,w)$ of $\mathcal N$ corresponds to a fully-labeled vertex-pair of $P\times Q'$, and hence it has 
a duplicate label (by $O_2$). Relaxing the inequality corresponding to the duplicate label in $P$ and $Q'$
separately, we get the edges $(\mathcal E_w,w)$ and $(v,\mathcal E_v)$ of $\mathcal N$ respectively. Clearly, these are the
only adjacent edges of the vertex $(v,w)$ in $\mathcal N$. 
Hence, in a component of $\mathcal N$, edges alternate between type $(v,\mathcal E_v)$ and
$(\mathcal E_w,w)$, and the degree of every vertex of $\mathcal N$ is exactly two.
Therefore, $\mathcal N$ consists of infinite paths and cycles on the $1$-skeleton of $P\times Q'$. 
Note that a path in $\mathcal N$ has unbounded edges on both the sides. 
Further, a component of $\mathcal N$ may be constructed by combining a component of $\mathcal N^P$ (say $\mathcal C$) and
the corresponding component of $\mathcal N^{Q'}$ ($\{\mathcal E_v \ |\ v\in \mathcal C\}$). 

Using the above analysis, we only need to show that there is exactly one path in $\mathcal N$ to prove Proposition \ref{pr1}.
Let the {\em support-pair} of a vertex $(y,\pi_2) \in P$ be $(I,J)$ where $I=\{i\in S_1\ |\ A_iy-\pi_2=0\}$ and $J=\{j \in
S_2\ |\ y_j>0\}$. Note that $|L(y,\pi_2)|=n$, hence $|I|=|J|$. 
Let $\beta_{j_s}=\min_{j \in S_2} \beta_j$, $i_s=\argmax_{i\in S_1} a_{ij_s}$, $\beta_{j_e}=\max_{j\in S_2}\beta_e$, and
$i_e=\argmax_{i\in S_1} a_{ij_e}$. In other words, the indices $j_s$ and $j_e$ correspond to the minimum and maximum
entries in $\beta$ respectively, and the indices $i_s$ and $i_e$ correspond to the maximum entry in $A^{j_s}$ and $A^{j_e}$
respectively. It is easy to see that $j_s\neq j_e$, since $Q'$ is non-degenerate. 

\begin{lemma}\label{le3}
There exist two vertices $v_s$ and $v_e$ in $P$, with support-pairs $(\{i_s\},\{j_s\})$ and $(\{i_e\},\{j_e\})$
respectively.  
\end{lemma}
\begin{proof}
Let $y \in \Delta_2$ be such that $y_{j_s}=1$ and $y_j=0,\ \forall j\neq j_s$. Clearly, $v_s=(y,a_{i_sj_s})\in P$ and
$|L(v_s)|=n$. Similarly, the vertex $v_e \in P$ may be obtained by setting $y_{j_e}=1$ and the remaining $y_j$s to zero.\qed
\end{proof}

Next we show that there are exactly two unbounded edges of type $(v,\mathcal E_v)$ in $\mathcal N$, all other edges have two
bounding vertices.
\newpage 

\begin{lemma}\label{sle4}
An edge $(v,\mathcal E_v)\in \mathcal N$ has exactly one bounding vertex if $v$ is either  $v_s$ or $v_e$, otherwise it
has two bounding vertices.
\end{lemma}
\begin{proof} 
Let $v=v_s$. The points in $\mathcal E_v$ satisfy
\begin{eqnarray}
x_{i_s}=1\mbox{ and }\forall i\neq i_s,\ x_i=0, \hspace{.2in} \pi_2=c_{i_sj_s}+\beta_{j_s}\lambda \label{eqx}\\
\forall j\neq j_s,\ c_{i_sj}+\beta_j\lambda \leq c_{i_sj_s}+\beta_{j_s}\lambda\hspace{1.5cm} \nonumber
\end{eqnarray}

Since $\beta_j\geq \beta_{j_s}$, we get $\lambda \leq \displaystyle\frac{c_{i_sj_s}-c_{i_sj}}{\beta_j-\beta_{j_s}}$.
Let $\lambda_s=\displaystyle\min_{j\neq j_s} \displaystyle\frac{c_{i_sj_s}-c_{i_sj}}{\beta_j-\beta_{j_s}}$, then
$\mathcal E_v=\{(x,\lambda,\pi_2)\ |\ \lambda\in (-\infty,\ \lambda_s],\ x \mbox{ and }\pi_2 \mbox{ satisfy (\ref{eqx})} 
\}$. Note that on $\mathcal E_v$, $x$ is a constant and $\lambda$ varies from $-\infty$ to $\lambda_s$. Moreover the point
corresponding to $\lambda=\lambda_s$ is a vertex, because one more inequality becomes tight there. 
Similarly for $v=v_e$, $\lambda$ varies from $\lambda_e=\displaystyle\max_{j\neq j_e}
\displaystyle\frac{c_{i_ej}-c_{i_ej_e}}{\beta_{j_e}-\beta_{j}}$ to $\infty$ on $\mathcal E_v$, and $\lambda=\lambda_e$
corresponds to  a vertex of $\mathcal E_{v}$.

Let a vertex $v \in P$ be such that $v\neq v_s$, $v\neq v_e$ and $\mathcal E_v\neq \emptyset$. We show that $\mathcal E_v$
has exactly two bounding vertices. Let $(I,J)$ be the support-pair corresponding to $v$. There are two cases.

\vspace{0.2cm}
\noindent {\bf Case 1} - $|I|=|J|=1$:
Let $I=\{i_1\}$ and $J=\{j_1\}$. Then for all the points in $\mathcal E_v$, $x_{i_1}=1$ and all other
$x_i$s are zero. Let $J_l=\{j \ |\ \beta_j<\beta_{j_1}\}$ and $J_g=\{j \ |\ \beta_j>\beta_{j_1}\}$. Clearly $j_s \in J_l$
and $j_e \in J_g$. All the points in $\mathcal E_v$ must satisfy the inequalities $c_{i_1j}+\beta_j\lambda\leq
c_{i_1j_1}+\beta_{j_1}\lambda,\ \forall j \notin J$, and using them, we get the following upper and lower bounds on
$\lambda$.  \[
\displaystyle\max_{j \in J_g}\displaystyle\frac{c_{i_1j_1}-c_{i_1j}}{\beta_{j}-\beta_{j_1}}\leq \lambda\leq
\displaystyle\max_{j \in J_l}\displaystyle\frac{c_{i_1j}-c_{i_1j_1}}{\beta_{j_1}-\beta_{j}}
\]
Therefore, the values of $\lambda$ on $\mathcal E_v$, form a closed and bounded interval, and for each extreme point of this
interval, there is a vertex in $\mathcal E_v$. 

\vspace{0.2cm}
\noindent {\bf Case 2} - $|I|=|J|>1$:
Note that exactly $m$ inequalities of $Q'$ are tight at $\mathcal E_v$ because $|L(v)|=n$ and $Q'$ is non-degenerate.
These $m$ tight inequalities with $\sum_{i=1}^m x_i=1$ form a $1$-dimensional line $L$ in the $(x,\lambda,\pi_2)$-space, and
clearly $\mathcal E_v=L\cap Q'$. Let $w=(x,\lambda,\pi_2)\in L$ and $d$ be a unit vector along the line $L$.  For a $w'\in
L$, there exists a unique $\epsilon \in \mathbb R$ such that $w'=w+\epsilon d$. Let $d(x_i)$ be the coordinate of $d$
corresponding to $x_i$. Note that $\sum_{i=1}^m d(x_i)=0$, because $L$ satisfies $\sum_{i=1}^m x_i=1$. Further $\exists i \in
I$ such that $d(x_i)\neq 0$, otherwise $x$ becomes constant on $L$, which in turn imply that $\lambda$ and $\pi_2$ are also
constants on $L$.  Hence $\exists i_1, i_2 \in I$ s.t.  $d(x_{i_1})>0$ and $d(x_{i_2})<0$. For all the points in $\mathcal
E_v$, the inequalities $x_i\geq 0,\ \forall i\in I$ hold. Using these, we get 
\[
x_{i_1}+\epsilon d(x_{i_1})\geq 0 \hspace{.1in}\Rightarrow\hspace{.1in} \epsilon \geq \frac{x_{i_1}}{d(x_{i_1})},\hspace{.2in}
x_{i_2}+\epsilon d(x_{i_2})\geq 0 \hspace{.1in}\Rightarrow\hspace{.1in} \epsilon \leq \frac{x_{i_2}}{d(x_{i_2})}
\]
From the above observations, we may easily deduce that the set $\{\epsilon\ |\ w+\epsilon d \in \mathcal E_v\}$ is a closed
and bounded interval $[b_l,\ b_u]$.  Moreover, at the extreme points $w_u=w+b_u d$ and $w_l=w+b_l d$ of $\mathcal E_v$, one
more inequality is tight. Therefore, $w_u$ and $w_l$ are the vertices in $\mathcal E_v$. 
\qed
\end{proof}

Now we are in a position to prove Proposition \ref{pr1}.
\\

\noindent{\bf Proof of Proposition \ref{pr1}}: \\
For a vertex $w=(x,\lambda,\pi_2) \in \mathcal N^{Q'}$, $\exists r\leq m$ such that $x_r>0$ since $\sum_{i=1}^m x_i=1$. In
that case, $A_ry=\pi_1$ holds on the corresponding edge $\mathcal E_w \in \mathcal N^P$ ($O_4$). This implies that the edge
$\mathcal E_w$ is bounded from both the sides, since $\forall j\in S_2,\ 0\leq y_j\leq 1$ and $A_{min}\leq \pi_1\leq A_{max}$
on the edge $\mathcal E_w$, where $A_{min}ٍ=\min_{_{(i,j)\in S_1\times S_2}} a_{ij}$ and
$A_{max}ٍ=\max_{_{(i,j)\in S_1\times S_2}} a_{ij}$.  Therefore, there are exactly two unbounded edges in the
set $\mathcal N$ namely $(v_s,\mathcal E_{v_s})$ and $(v_e,\mathcal E_{v_e})$ (Lemma \ref{sle4}).  This proves that $\mathcal
N$ contains exactly one path $\mathcal P$, with unbounded edges $(v_s,\mathcal E_{v_s})$ and $(v_e,\mathcal E_{v_e})$ at both
the ends. All the other components of $\mathcal N$ form cycles ($\mathcal C_i$s). 
\qed
\vspace{10pt}

From Proposition \ref{pr1}, it is clear that $\mathcal N$ contains at least the path $\mathcal P$. We show the importance of
$\mathcal P$ in the next two lemmas.

\begin{lemma}\label{le_lambda}
For every $a \in \mathbb R$, there exists a point $((y,\pi_1),(x,\lambda,\pi_2))\in \mathcal P$ such that $\lambda=a$.
\end{lemma}
\begin{proof}
Since $\mathcal P$ is a continuous path in $P\times Q'$ (Proposition \ref{pr1}), therefore $\lambda$ changes
continuously on $\mathcal P$. Moreover, in the proof of Lemma \ref{sle4}, we saw that on the edge $(v_s,\mathcal E_{v_s})\in
\mathcal P$, $\lambda$ varies from $-\infty$ to $\lambda_s$ and on the edge $(v_e,\mathcal E_{v_e})\in \mathcal P$ it varies
from $\lambda_e$ to $\infty$. Therefore for any $a \in \mathbb R$, there is a point $((y,\pi_1),(x,\lambda,\pi_2))$ in
$\mathcal P$ such that $\lambda=a$.  \qed 
\end{proof}

Consider a game $\alpha \in \Gamma$, and the corresponding hyper-plane $H\equiv\lambda-\sum_{i=1}^m \alpha_i x_i=0$. 
Note that, every point in $\mathcal N \cap H$ corresponds to a NESP of the game $G(\alpha)$ and vice-versa. 

\begin{lemma}\label{le_cover}
The path $\mathcal P$ of $\mathcal N$ covers at least one NESP of the game $G(\alpha)$.
\end{lemma}
\begin{proof}
If there are points in $\mathcal P$ on
opposite sides of $H$, then the set $\mathcal P \cap H$ has to be non-empty. Let $w_1=(x^1,\lambda_1,\pi_2^1) \in
\mathcal P$ and $w_2=(x^2,\lambda_2,\pi_2^2) \in \mathcal P$ be such that $\lambda_1=\min_{i\in S_1} \alpha_i$ and
$\lambda_2=\max_{i\in S_1} \alpha_i$. Note that $w_1$ and $w_2$ exist (Lemma \ref{le_lambda}) and they satisfy
$\lambda_1-\sum_{i=1}^m \alpha_i x^1_i\leq 0$ and $\lambda_2-\sum_{i=1}^m \alpha_i x^2_i\geq 0$. \qed
\end{proof}

\begin{remark}
The proof of Lemma \ref{le_cover} in fact shows the existence of a Nash equilibrium for a bimatrix game. It is also easy
to deduce that the number of Nash equilibria of a non-degenerate bimatrix game is odd from the fact that $\mathcal N$
contains a set of cycles and a path (Proposition \ref{pr1}), simply because a cycle must intersect the hyper-plane $H$ an even
number of times, and the path must intersect $H$ an odd number of times.
\end{remark}

From the proof of Proposition \ref{pr1}, it is clear that every vertex of $\mathcal N$ has a duplicate label and the two
edges incident on a vertex may be easily obtained by relaxing the inequality corresponding to the duplicate
label in $P$ and in $Q'$. Therefore, given a point of some component of $\mathcal N$, it is easy to trace the full component
by leaving the duplicate label in $P$ and in $Q'$ alternately at every vertex. Using this fact along with Lemma \ref{le_cover},
we design an algorithm to find a Nash equilibrium of a bimatrix game in Section \ref{enum}. 
For the moment, we show that the edges of $\mathcal N$ may be easily oriented. 

Consider a vertex $u=(v,w)\in \mathcal N$, where $v=(y,\pi_1)$, and $w=(x,\lambda,\pi_2)$. 
Let $X=\{i \in S_1\ |\ A_iy=\pi_1\}$, and $Y=\{j \in S_2\ |\ x^TC^j+\beta_j\lambda=\pi_2\}$ be ordered sets. 
Note that $X=L(v)\cap S_1$ and $Y=L(w)\cap S_2$. 
Let $-X=S_1\setminus X$, and $-Y=S_2\setminus Y$ be the complements. The duplicate label, say $l$, of the vertex $u$ is
either in $X$ or in $Y$. Let $l \in X$, {\em i.e.,} $A_ly=\pi_1$ and $x_l=0$ hold at $u$. Let $-I_k$ be the negative of the
$k\times k$ identity matrix, $A_{_X}^{^Y}=[a_{ij}]_{i\in X,j \in Y}$ be the submatrix of $A$ and similarly $\beta_{_Y}$ be
the subvector. The set of tight inequalities at $v$ and $w$ may be written as follows:
\begin{eqnarray}
\left[\begin{array}{ccc}
\multicolumn{2}{c}{1_{1\times n}} & 0\vspace{0.05cm}\\
A_{_X}^{^Y} & A_{_X}^{^{-Y}} & -1_{_{|X|\times 1}} \vspace{0.05cm}\\
0_{_{|-Y|\times |Y|}} & {-I_{_{|-Y|}}} & 0_{_{|-Y|\times 1}}
\end{array}\right]
\left[\begin{array}{c}
y_{_Y}\vspace{0.05cm} \\ y_{_{-Y}}\vspace{0.05cm} \\ \pi_1
\end{array}\right] = 
\left[\begin{array}{c}
1\vspace{0.05cm}\\
0_{_{|X|\times 1}}\vspace{0.05cm} \\
0_{_{|-Y|\times 1}}
\end{array}\right] \label{eq_v} \\\nonumber\\
\left[\lambda\ x_{_X}\ x_{_{-X}}\ \pi_2\right]
\left[\begin{array}{cccc}
0 & \beta_{_Y} & 0 & 0_{_{1\times |-X|}} \vspace{0.05cm}\\
\multirow{2}{*}{$1_{m \times 1}$} & C_{_X}^{^Y} & \multirow{2}{*}{$-e_l$} & \hspace{0.1cm} 0_{_{|X|\times|-X|}} \vspace{0.05cm}\\
& C_{_{-X}}^{^Y} & &-I_{_{|-X|}}\vspace{0.05cm}\\
0 & -1_{_{1\times |Y|}} & 0 & 0_{_{1 \times |-X|}}
\end{array}\right]=
\left[1\ 0_{_{1\times |Y|}}\ 0\ 0_{_{1\times |-X|}}\right] \label{eq_w}
\end{eqnarray}

In the above expression $-e_l$ is a negative unit vector of size $m$, with $-1$ in the position corresponding to $x_l$. 
Let the matrices of (\ref{eq_v}) and (\ref{eq_w}) be denoted by $E(v)$ and $E(w)$ respectively. 
For the case $l \in Y$, $E(v)$ and $E(w)$ may be analogously defined.
It is easy to see that the coefficient matrix of tight equations at $u$ may be written as $E(u)=\left[\begin{array}{cc} E(v) & 0 \\ 0 &
E(w)^T\end{array}\right]$. Using $det(E(u))=det(E(v))*det(E(w))$, we define the sign of vertex $u$ as follows: 
\[
s(u)=sign(det(E(u)))
\]

Note that, since $E(v)$ and $E(w)$ are well-defined, $s$ is a well-defined function. Using the function $s$ on the
vertices of $\mathcal N$, next we give direction to the edges of $\mathcal N$. 

\begin{lemma}\label{le_dir}
Let $E$ be the set of edges of $\mathcal N$, and $E'=\{\overrightarrow{u,u'},\ \overleftarrow{u,u'}\ |\ \overline{u,u'}\in
\mathcal N\}$ be the set of directed edges. There exists a (efficiently computable) function $\rightarrow:E\rightarrow E'$
such that it maps a cycle of $\mathcal N$ to a directed cycle and the path $\mathcal P$ to a path oriented from
$(v_s,\mathcal E_{v_s})$ to $(v_e,\mathcal E_{v_e})$.
\end{lemma}
\begin{proof}

Define the function $\rightarrow$ as follows: Let $u=(v,w)$ be a vertex of $\mathcal N$ and let $u_p=(v',w)$ and $u_q=(v,w')$
be it's adjacent vertices obtained by relaxing the inequalities corresponding to it's duplicate label in $P$ and in $Q'$
respectively. If $s(u)=+1$ then $\rightarrow(\overline{u,u_p})=\overrightarrow{u,u_p}$ and
$\rightarrow(\overline{u,u_q})=\overleftarrow{u,u_q}$, otherwise $\rightarrow(\overline{u,u_p})=\overleftarrow{u,u_p}$ and
$\rightarrow(\overline{u,u_q})=\overrightarrow{u,u_q}$.
In other words, if $s(u)=+1$ then direct the edges $(\mathcal E_w,w)$ and $(v,\mathcal E_v)$ away from $u$ and
towards $u$ respectively, otherwise give opposite directions. 

Note that, $\rightarrow(\overline{u,u'})$ is polynomially computable using any of the $s(u)$ and $s(u')$. Further, in order
to prove the consistency of function $\rightarrow$, we need to show that $u$ and $u'$ have opposite signs. 

\begin{claim}
Let $u$ and $u'$ be adjacent vertices of $\mathcal N$, then $s(u)*s(u')=-1$, {\em i.e.}, $s(u)$ and $s(u')$
are opposite. 
\end{claim}
\begin{proof}
Let $u=(v,w)$. The proof may be easily deduced from the following facts:
\begin{itemize}
\item In a polytope, the coefficient matrix of tight inequalities for the adjacent vertices have determinants of opposite
signs if they are the same except for the row which has been exchanged \cite{shap}. 
\item Vertices $u$ and $u'$ are fully-labeled and both have a duplicate label. Further, to obtain $u'$ from $u$, the
inequality corresponding to it's duplicate label should be relaxed in $P$ or $Q'$. 
\item Reordering of the elements in set $X$ ($Y$) does not change $det(E(u))$, since it enforces a reordering of the 
corresponding rows (columns) in both $E(v)$ and $E(w)$.
\item Reordering of the elements in set $-X$ does not change $det(E(u))$, since in $E(w)$, the columns
corresponding to the equations of type $x_i=0$ (except the one with the duplicate label) 
should be written such that they form $-I_{_{|-X|}}$. Similarly, reordering of the elements in set $-Y$ does not change
$det(E(u))$. 
\qed
\end{itemize}
\end{proof}

Clearly, the function $\rightarrow$ maps a cycle of $\mathcal N$ to a directed cycle. 
Therefore, we get the directed traversal of a component of $\mathcal N$ by leaving the duplicate label in $P$ if the current
vertex has positive sign otherwise leaving the duplicate label in $Q'$. Further, it is easy to check that the sign associated
with the vertex of the extreme edge $(v_s,\mathcal E_{v_s})$ is positive. Therefore, the path $\mathcal P$ gets oriented from
$(v_s,\mathcal E_{v_s})$ to $(v_e,\mathcal E_{v_e})$.  \qed
\end{proof}

The direction of the edges of $\mathcal N$, defined by function $\rightarrow$, may be used to determine the index of every
Nash equilibrium for a game in $\Gamma$. 
The definition of index requires the game to be non-negative, {\em i.e.,} $A>0,\ B>0$ \cite{stengel}. Note that if a game is
not non-negative, then it may be modified to an equivalent non-negative game by adding a positive constant to its payoff
matrices. Let $(x, y)$ be a NESP of a non-degenerate non-negative bimatrix game $(A, B)$. Let $I=\{i \in S_2\ |\ x_i>0\}$ and
$J=\{j \in S_1\ |\ y_j>0\}$ with the corresponding submatrices $A_{_I}^{^J}$ and $B_{_I}^{^J}$ of the payoff matrices $A$ and
$B$. Then the index of $(x, y)$ is defined as
\[
(-1)^{|I|+1} sign(det(A_{_I}^{^J})*det(B_{_I}^{^J}))
\]

Let $\alpha \in \Gamma$ be a non-degenerate non-negative game and let $H: \lambda-\sum_{i=1}^m \alpha_i x_i$, $H^-$ and $H^+$
be the corresponding hyper-plane and half-spaces. 

\begin{proposition}\label{pr2}
Let an edge $\overline{u,u'} \in \mathcal N$ intersect $H$ at a NESP $(x,y)$ of $G(\alpha)$, and let
$\rightarrow(\overline{u,u'})=\overrightarrow{u,u'}$. If $u \in H^-$ and $u' \in H^+$ then the index of $(x,y)$ is $+1$,
otherwise it is $-1$. 
\end{proposition}
\begin{proof}
Since $\overline{u,u'}$ intersects $H$ and the coordinates of $y$ and $\pi_1$ are zero in $H$, the edge is of
type $(v,\mathcal E_v)$. Therefore, let $\overline{u,u'}=\overline{(v,w),(v,w')}$. Clearly, $s(u)=-1$ since
$\rightarrow(\overline{u,u'})=\overrightarrow{u,u'}$. Let the ordered sets $X$, $Y$ and their complements be as defined above. Let
$l$ be the duplicate label of $u$. Clearly, either $l \in Y$ or $l \in X$. 

Suppose $l \in Y$. Let $d$ be the direction obtained by relaxing the inequality $l$ at $u$ in $Q'$, which leads to the vertex
$u'$. The dot product of $d$ with the normal vector of $H$ may be obtained by replacing the column corresponding to
$x^TC^l+\beta_l\lambda-\pi_2=0$ with the normal vector in $E(w)$. If $u \in H^-$ and $u' \in H^+$ then this dot product is
positive, otherwise it is negative. Next we show that the expression of the dot product may be simplified to match with the
expression of the index of $(x,y)$. 

Let the payoff matrix of the column player in $G(\alpha)$ be denoted by $B$, {\em i.e.}, $B=C+\alpha\cdot\beta^T$ and let
$a=x^TAy$ and $b=x^TBy$. Since $A>0$ and $B>0$, $a$ and $b$ are positive. Clearly, the sets $I$ and $J$ associated with the NESP
$(x,y)$ are such that $I=X$ and $J=Y\setminus {l}$. Let $k=|I|=|J|$. We reorder the elements in set $Y$ such that $Y=[J\ l]$.
Note that this does not change the sign of $det(E(u))$, since it forces the similar reordering of the columns of both $E(v)$
and $E(w)$. The expression for the dot product is:
\[
\begin{array}{cl}
& \displaystyle\frac{-1}{det(E(u))} * det
\left[\begin{array}{cccc}
\multicolumn{3}{c}{1_{_{1\times n}}} & 0 \vspace{0.05cm}\\
A_{_X}^{^J} & A_{_X}^{^l}& A_{_X}^{^{-Y}} & -1_{_{|X|\times 1}} \vspace{0.05cm}\\
0_{_{1\times k}}& -1 & 0_{_{1\times |-Y|}} & 0 \vspace{0.05cm}\\
\multicolumn{2}{c}{0_{_{|-Y|\times (k+1)}}} & -I_{_{|-Y|}} & 0_{_{|-Y|\times 1}}
\end{array}\right] * det
\left[\begin{array}{cccc}
0 & \beta_{_J} & 1 & 0_{_{1\times |-X|}} \vspace{0.05cm}\\
\multirow{2}{*}{$1_{m \times 1}$} & C_{_X}^{^J} & -\alpha_{_X}  & 0_{_{|X|\times |-X|}}\vspace{0.05cm}\\
& C_{_{-X}}^{^J} & -\alpha_{_{-X}} & -I_{_{|-X|}}\vspace{0.05cm}\\
0 & -1_{_{1\times |J|}} & 0 & 0_{_{1 \times |-X|}}
\end{array}\right]  \end{array}
\]
Since $I=X$ and $|I|=|J|=k$, we get,
\[
\begin{array}{r}
\displaystyle\frac{(-1)^{1+(n-k)(2k+4)}}{det(E(u))} * det
\left[\begin{array}{cc}
1_{_{1\times k}} & 0 \vspace{0.05cm}\\
A_{_I}^{^J} & -1_{_{k\times 1}} \\
\end{array}\right] * (-1)^{(m-k)(2k+6)} * det
\left[\begin{array}{ccc}
0 & \beta_{_J} & 1  \vspace{0.05cm}\\
1_{_{k \times 1}} & C_{_I}^{^J} & -\alpha_{_I} \vspace{0.05cm} \\
0 & -1_{_{1\times k}} & 0 
\end{array}\right] \\ \\
=\displaystyle\frac{-1}{det(E(u))} * det
\left[\begin{array}{cc}
1_{_{1\times k}} & 0 \vspace{0.05cm}\\
A_{_I}^{^J} & -1_{_{k\times 1}} \\
\end{array}\right] * det
\left[\begin{array}{ccc}
0 & 0_{1\times k} & 1  \vspace{0.05cm}\\
1_{_{k \times 1}} & (C+\alpha\cdot\beta^T)_{_I}^{^J} & -\alpha_{_I} \vspace{0.05cm} \\
0 & -1_{_{1\times k}} & 0 
\end{array}\right] \\ \\
\end{array}
\]
Since $B=C+\alpha\cdot\beta^T$, $A_{_I}^{^J}y_{_J}=a*1_{_{|I|\times 1}}$ and 
$x_{_I}^{T}B^{^J}_{_I}=b*1_{_{1\times |J|}}$, we get,
\[
\begin{array}{lcl}
 \displaystyle\frac{-1}{det(E(u))} * det
\left[\begin{array}{cc}
1_{_{1\times k}} & \frac{1}{a} \vspace{0.05cm}\\
A_{_I}^{^J} & 0_{_{k\times 1}} \\
\end{array}\right] * (-1)^{k+3}\left[\begin{array}{cc}
1_{_{k \times 1}} & B_{_I}^{^J}  \vspace{0.05cm} \\
\frac{1}{b} & 0_{_{1\times k}}  
\end{array}\right] 
&=& \displaystyle\frac{(-1)^{k+2(k+2)}}{det(E(u))*a*b} * det(A_{_I}^{^J}) * det (B_{_I}^{^J})\\  
&=& \displaystyle\frac{(-1)^{k}}{det(E(u))*a*b} * det(A_{_I}^{^J}) * det (B_{_I}^{^J})  
\end{array}
\]\vspace{0.15cm}

When the duplicate label $l$ is in $X$, we may derive the same expression for the dot product by similar reductions.
Since $s(u)=sign(det(E(u)))=-1$, $a>0$ and $b>0$, the sign of the above expression is same as the index of $(x,y)$. \qed
\end{proof}

From Proposition \ref{pr2}, it is easy to see that in a component, the index of the Nash equilibria alternates\footnote{The
two endpoints of a LH path also have opposite index \cite{shap}.}. Further, both the first and the last Nash equilibria, on
the path $\mathcal P$, have index $+1$. This proves that the number of Nash equilibria with index $+1$ is one more than the
number of Nash equilibria with index $-1$, which is an important known result \cite{stengel,shap}. 

Recall that $\mathcal N$ surely contains the path $\mathcal P$ and in addition it may also contain some cycles. From Lemma
\ref{le_con}, it is clear that if $\mathcal N$ is disconnected, then $E_\Gamma$ is also disconnected.
Example \ref{ex1} shows that $E_\Gamma$ may be disconnected in general by illustrating a disconnected $\mathcal N$ ({\em
i.e.}, $\mathcal N$ with a cycle). For a more detailed structural description of $E_\Gamma$, we refer the reader to Appendix
\ref{region}.

\begin{example}\label{ex1}
Consider the following $A$, $C$ and $\beta$.
\[
A=\left[\begin{array}{ccc}
0\ & 9\ & 9\\
6\ & 6\ & 5 \\
9\ & 7\ & 2
\end{array}\right],\hspace{.3in}
C=\left[\begin{array}{ccc}
6\ & 8\ & 6 \\
5\ & 8\ & 8 \\
4\ & 3\ & 0
\end{array}\right],\hspace{.3in} 
\beta=\left[\begin{array}{c}
9 \\
7 \\
8
\end{array}\right].
\]
The set $\mathcal N$ of the corresponding game space $\Gamma$ contains a path $\mathcal P$ and a cycle $\mathcal C_1$. From
Proposition \ref{pr1}, it is clear that a component of $\mathcal N$ may be obtained from a component of $\mathcal N^P$ and the
corresponding component of $\mathcal N^{Q'}$. Therefore we demonstrate the path $\mathcal P^P$ and the cycle $\mathcal C_1^P$
of $\mathcal N^P$, and using them $\mathcal P$ and $\mathcal C_1$ of $\mathcal N$ may be easily obtained. The path $\mathcal
P^P$ is $\overline{v_s,v_1},\overline{v_1,v_e}$, where $v_s=((0,\ 1,\ 0),\ 9), v_1=((0.18,\ 0.82,\ 0),\ 7.36)$ and $v_e=((1,\
0,\ 0),\ 9)$.  The cycle $\mathcal C_1^P$ is $\overline{v_2,v_3},\overline{v_3,v_4},\overline{v_4,v_2}$, where $v_2=((0.5,\
0,\ 0.5),\ 5.5),v_3=((0.38,\ 0.18,\ 0.44),\ 5.56)$ and $v_4=((0.4,\ 0,\ 0.6),\ 5.4)$. Note that $v_s$ and $v_e$ correspond
to the minimum and maximum $\beta_j$ respectively (Lemma \ref{sle4}).
\qed
\end{example}

Since $\Gamma$ ($\equiv \mathbb R^m$) is connected, hence if $E_\Gamma$ is disconnected then it is not homeomorphic to
$\Gamma$. 

\section{Rank-1 Space and Homeomorphism}\label{hom}
From the discussion of the last section, we know that $\Gamma$ and $E_\Gamma$ are not homeomorphic in general (illustrated by 
Example \ref{ex1}). Surprisingly, they turn out to be homeomorphic if $\Gamma$ consists of only rank-$1$
games, {\em i.e.}, $C=-A$. Recall that $E_\Gamma$ forms a single connected component iff $\mathcal N$ has only one component
(Lemma \ref{le_con}).  First we show that when $C=-A$, the set $\mathcal N$ consists of only a path.

For a given matrix $A \in \mathbb R^{mn}$ and a vector $\beta \in \mathbb R^{n}$, we fix the game space to
$\Gamma=\{(A,-A+\alpha\cdot\beta^T)\ |\ \alpha \in \mathbb R^m\}$. Without loss of generality (wlog) we assume that $A$ and
$\beta$ are non-zero and the corresponding polytopes $P$ and $Q'$ are non-degenerate. Lemma \ref{ho_le1} shows that the set $\mathcal
N$ may be easily identified on the polytope $P\times Q'$.

\begin{lemma}\label{ho_le1}
For all $(v,w)=((y,\pi_1),(x,\lambda,\pi_2))$ in $P\times Q'$, we have $\lambda(\beta^T\cdot y) -\pi_1 - \pi_2\leq0$, and the
equality holds iff $(v,w) \in \mathcal N$.
\end{lemma}
\begin{proof} Recall that $C=-A$, hence from (\ref{eq2}) and (\ref{eq3}), we get $x^T\cdot (A\cdot y- \pi_1)\leq 0$ and
$(x^T\cdot(-A)+\beta^T\lambda-\pi_2)\cdot y\leq 0$. By summing up these two inequalities, we get $\lambda(\beta^T\cdot
y)-\pi_1-\pi_2\leq 0$.  
If $(v,w)\in \mathcal N$, then $\forall i\leq m,\ x_i>0 \Rightarrow A_i\cdot y-\pi_1=0$ and $\forall j\leq n,\ y_j>0
\Rightarrow x^T(-A^j)+\beta_j\lambda-\pi_2=0$, hence $\lambda(\beta^T\cdot y)-\pi_1-\pi_2=0$. 

If $(v,w)\notin \mathcal N$, then at least one label $1\leq r\leq m+n$ is missing from $L(v)\cup L(w)$. Let $r\leq m$ (wlog),
then $x_r>0$ and $A_r\cdot y-\pi_1<0$, which imply that $x^T\cdot (A \cdot y-\pi_1)<0$. Therefore, $\lambda(\beta^T\cdot
y)-\pi_1-\pi_2<0$. \qed
\end{proof}

Motivated by the above lemma, we define the following parametrized linear program $LP(\delta)$.
\begin{eqnarray}\label{eq_lp}
\begin{array}{ll}
\vspace{2pt}
LP(\delta):\hspace{.1in}  & max\ \ \delta(\beta^T\cdot y)-\pi_1-\pi_2 \\
& \hspace{.4in}(y,\pi_1)\in P \\
& \hspace{.4in}(x,\lambda,\pi_2) \in Q' \\
& \hspace{.4in}\lambda=\delta
\end{array}
\end{eqnarray}

Note that the above linear program may be broken into a parametrized primal linear program and it's dual, with
$\delta$ being the parameter. The primal may be defined on polytope $P$ with the cost function {\em maximize:
}$\delta(\beta^T\cdot y)-\pi_1$ and it's dual is on polytope $Q'$ with additional constraint $\lambda=\delta$ and the cost
function {\em minimize: }$\pi_2$. 
\begin{remark}
$LP(\delta)$ may look similar to the parametrized linear program, say $TLP(\xi)$, by Theobald \cite{the}. However the
key difference is that $TLP(\xi)$ is defined on the best response polytopes of a given game ({\em i.e.}, $P(\alpha)\times
Q(\alpha)$ for the game $G(\alpha)$), while $LP(\delta)$ is defined on a bigger polytope ($P\times Q'$) encompassing best
response polytopes of all the games in $\Gamma$. A detailed comparison is given in Section \ref{enum}.
\end{remark}

Let $OPT(\delta)$ be the set of optimal points of $LP(\delta)$. In the next lemma, we show that $\forall \delta\in\mathbb R$,
$OPT(\delta)$ is exactly the set of points in $\mathcal N$, where $\lambda=\delta$.

\begin{lemma}\label{ho_le2}
$\forall a \in \mathbb R$, $OPT(a)=\{((y,\pi_1),(x,\lambda,\pi_2))\in \mathcal N\ |\ \lambda=a\}$ and $OPT(a)\neq\emptyset$.
\end{lemma}
\begin{proof}Clearly the feasible set of $LP(a)$ consists of all the points of $P\times Q'$, where $\lambda=a$. Therefore
the set $\{((y,\pi_1),(x,\lambda,\pi_2))\in \mathcal N\ |\ \lambda=a\}$ is a subset of the feasible set of $LP(a)$. The set
$\{((y,\pi_1),(x,\lambda,\pi_2))\in \mathcal N\ |\ \lambda=a\}$ is non-empty (Lemma \ref{le_lambda}).
From Lemma \ref{ho_le1}, it is clear that the maximum possible value, the cost function of $LP(a)$ may achieve is $0$,
and it is achieved only at the points of $\mathcal N$. Therefore,
$OPT(a)=\{((y,\pi_1),(x,\lambda,\pi_2))\in \mathcal N\ |\ \lambda=a\}$ and $OPT(a)\neq\emptyset$.  \qed
\end{proof}

Lemma \ref{ho_le2} implies that for any $a \in \mathbb R$, the set $OPT(a)$ is contained in $\mathcal N$. 
Using this, next we show that $\mathcal N$ in fact consists of only one component. 

\begin{proposition}\label{ho_pr1}
$\mathcal N$ does not contain cycles.
\end{proposition}
\begin{proof}
Since $\mathcal N$ consists of a set of edges and vertices and $OPT(a)$ is a convex set, therefore $OPT(a)$ is contained in a
single edge of $\mathcal N$ (Lemma \ref{ho_le2}). 
From Proposition \ref{pr1}, it is clear that there is a path $\mathcal P$ in the set $\mathcal N$. Further, Lemma
\ref{le_lambda} shows that for every $a \in \mathbb R$, there exists a point $((y,\pi_1),(x,\lambda,\pi_2))\in\mathcal P$,
where $\lambda=a$. It implies that $OPT(a),\forall a\in\mathbb R$ is contained in the path $\mathcal P$.
Therefore there is no other component in $\mathcal N$.  \qed
\end{proof}

From Proposition \ref{ho_pr1}, it is clear that $\mathcal N$ consists of only the path $\mathcal P$, henceforth we refer to
$\mathcal N$ as a path. To construct homeomorphism maps between $E_\Gamma$ and $\Gamma$, we need to encode a point
$(\alpha,x,y) \in E_\Gamma$ (of size $2m+n$) into a vector $\alpha'\in \Gamma$ (of size $m$), such that $\alpha'$ uniquely
identifies the point $(\alpha,x,y)$ ({\it i.e.,} a bijection). Recall that for every point in $E_\Gamma$, there is a unique
point on the path $\mathcal N$ (Lemma \ref{le1}). Therefore, first we show that there is a bijection between $\mathcal N$
and $\mathbb R$ and using this, we derive a bijection between $\Gamma$ and $E_{\Gamma}$. Consider the function $g:\mathcal
N\rightarrow \mathbb R$ such that 
\begin{eqnarray}
g((y,\pi_1),(x,\lambda,\pi_2))=\beta^T\cdot y+\lambda 
\end{eqnarray}

\begin{lemma}\label{ho_le3}
Each term of $g$, namely $\beta^T\cdot y$ and $\lambda$, monotonically increases on the directed path $\mathcal N$, and
the function $g$ strictly increases on it.
\end{lemma}
\begin{proof}
From the proof of Proposition
\ref{pr1}, we know that the edges of type $(v,\mathcal E_v)$ (where $v\in \mathcal N^P$ is a vertex) and of type $(\mathcal
E_w,w)$ (where $w \in \mathcal N^{Q'}$ is a vertex) alternate in $\mathcal N$. Clearly $\beta^T\cdot y$ is a constant on an
edge of type $(v,\mathcal E_v)$ and $\lambda$ is a constant on an edge of type $(\mathcal E_w,w)$. Now, consider the two
consecutive edges $(\mathcal E_{w},w)$ and $(v,\mathcal E_{v})$, where $\mathcal E_{w}=\overline{v',v}$ and $\mathcal
E_{v}=\overline{w,w'}$. It is enough to show that $\lambda$ and $\beta^T\cdot y$ are not constants on $(v,\mathcal
E_v)$ and $(\mathcal E_w,w)$ respectively, and $\beta^T\cdot y$ increases from $(v',w)$ to $(v,w)$ ({\it i.e.,} on
$(v,\mathcal E_v)$) iff $\lambda$ also increases from $(v,w)$ to $(v,w')$ ({\it i.e.,} on $(\mathcal E_w,w)$). 

Let $w=(x,\lambda,\pi_2),\ w'=(x',\lambda',\pi'_2),\ v=(y,\pi_1)$ and $v'=(y',\pi'_1)$.  Clearly, $OPT(\lambda)=(\mathcal
E_{w},w)$ and $(v,w')\in OPT(\lambda')$ (Lemma \ref{ho_le2}). Further $\lambda\neq \lambda'$, since $OPT(\lambda)$ contains
only one edge.

\begin{claim}
$\beta^T\cdot y'\neq \beta^T\cdot y$, and $\beta^T\cdot y'< \beta^T\cdot y\Leftrightarrow \lambda<\lambda'$.
\end{claim}

\begin{proof}
Since the feasible set of $LP(\lambda')$ contains all the points of $P\times Q'$ with $\lambda=\lambda'$, the point $(v',w')$
is a feasible point of $LP(\lambda')$. Note that $(v',w')$ is a suboptimal point of $LP(\lambda')$ otherwise $\mathcal
E_{w'}=\overline{v,v'}$ and $\mathcal E_{v'}=\overline{w',w}$, which creates a cycle in $\mathcal N$. Further, $(v,w') \in
OPT(\lambda')$, hence $\lambda'(\beta^T\cdot y)-\pi_1-\pi'_2 > \lambda'(\beta^T\cdot y')-\pi_1'-\pi'_2$. Since both $(v',w)$
and $(v,w)$ are in $OPT(\lambda)$, we get $\lambda(\beta^T\cdot y')-\pi_1'-\pi_2 = \lambda(\beta^T\cdot y)-\pi_1-\pi_2$.  
Summing up these two, we get $\lambda(\beta^T\cdot y')+\lambda'(\beta^T\cdot y)>\lambda(\beta^T\cdot y)+\lambda'(\beta^T\cdot
y') \Rightarrow (\beta^T\cdot y-\beta^T\cdot y')(\lambda'-\lambda)>0$. \qed
\end{proof}

From the above claim, it is clear that $\beta^T\cdot y$ is strictly monotonic on $(\mathcal E_{w},w)$ and $\lambda$ is
strictly monotonic on $(v,\mathcal E_{v})$. Further, if $\beta^T\cdot y$ increases on $(\mathcal E_{w},w)$ from $(v',w)$
to $(v,w)$ then $\lambda$ increases on $(v,\mathcal E_{v})$ from $(v,w)$ to $(v,w')$ and vice-versa.

Recall that on the directed path $\mathcal N$, $(v_s,\mathcal E_{v_s})$ is the first edge and $(v_e,\mathcal E_{v_e})$ is the
last edge (Lemma \ref{le_dir}).
Further, $\lambda$ varies from $-\infty$ to $\lambda_s$ on the first edge $(v_s,\mathcal E_{v_s})$, and it varies from
$\lambda_e$ to $\infty$ on the last edge $(v_e,\mathcal E_{v_e})$ (proof of Lemma \ref{sle4}). Therefore, $\lambda$ and
$\beta^T\cdot y$ increase monotonically on the directed path $\mathcal N$, and in turn $g$ strictly increases from $-\infty$
to $\infty$ on the path. 
\qed 
\end{proof}

Lemma \ref{ho_le3} implies that $g$ is a continuous, bijective function with a continuous inverse $g^{-1}:\mathbb
R\rightarrow \mathcal N$. Now consider the following candidate function $f:E_\Gamma \rightarrow \Gamma$ for the homeomorphism
map.
\begin{eqnarray}\label{eq4}
f(\alpha,x,y)=(\beta^T\cdot
y+\alpha^T\cdot x,\ \alpha_2-\alpha_1,\dots,\alpha_m-\alpha_1)^T 
\end{eqnarray}

Using the properties of $g$, next we show that $f$ indeed establishes a homeomorphism between $\Gamma$ and $E_\Gamma$.

\begin{theorem}
$E_\Gamma$ is homeomorphic to $\Gamma$.
\end{theorem}
\begin{proof}
The function $f$ of (\ref{eq4}) is continuous because it is a quadratic function. 
Further, we show that it is bijective. 
\begin{claim}
$f$ is a bijective function.
\end{claim}
\begin{proof}
We prove this by illustrating an inverse function $f^{-1}:\Gamma\rightarrow E_\Gamma$. Given $\alpha' \in \Gamma$, let
$(v,w)=((y,\pi_1),(x,\lambda,\pi_2))=g^{-1}(\alpha'_1)$ be the corresponding point in $\mathcal N$. This gives the values of
$x$, $y$ and $\lambda$. Using these values, we solve the following equalities with the variable vector ${\bs
a}=(a_1,\dots,a_m)$.  
\begin{eqnarray}
\forall i>1,\ a_i=\alpha'_i+a_1 \label{eq5} \\
\sum_{i=1}^m x_i a_i=\alpha'_1-\beta^T\cdot y \label{eq6}
\end{eqnarray}
It is easy to see that the above equations have a unique solution, which gives a unique value for the vector ${\bs a}$ 
and a unique point $({\bs a},x,y)\in E_\Gamma$. Clearly, $f({\bs a},x,y)=\alpha'$. 
\qed
\end{proof}
The inverse map $f^{-1}$ illustrated in the proof of the above claim is also continuous. The continuous
maps $f$ and $f^{-1}$ establish the homeomorphism between $E_\Gamma$ and $\Gamma$.\qed
\end{proof}

\section{Algorithms}\label{algo}
In this section, we present two algorithms to find Nash equilibria of a rank-1 game using the structure and monotonicity of
$\mathcal N$. First we discuss a polynomial time algorithm to find a Nash equilibrium of a non-degenerate rank-1 game. It
does a binary search on $\mathcal N$ using the monotonicity of $\lambda$. Later we give a path-following
algorithm which enumerates all Nash equilibria of a rank-1 game, and finds at least one for any bimatrix game (Lemma
\ref{le_cover}).

Recall that the best response polytopes $P$ and $Q$ (of (\ref{eq2})) of a non-degenerate game are non-degenerate,
and hence it's Nash equilibrium set is finite. Consider a non-degenerate rank-1 bimatrix game $(A,B)\in \mathbb R^{2mn}$
such that $A+B=\gamma\cdot\beta^T$, where $\gamma \in \mathbb R^m$ and $\beta \in \mathbb R^n$. We assume that $\beta$ is a
non-zero and non-constant\footnote{If $\beta$ is a constant vector, then the game $(A,B)$ may be converted into a zero-sum
game without changing it's Nash equilibrium set, by adding constants in the columns and rows of $A$ and $B$ respectively.}
vector, and both $A$ and $B$ are rational matrices. Let $c$ be the LCM of the denominators of the $a_{ij}$s, $\beta_i$s and
$\gamma_i$s.  Note that multiplying both $A$ and $B$ by $c^2$ makes $A$, $\gamma$ and $\beta$ integers, and the total bit
length of the input gets multiplied by at most $O(m^2n^2)$, which is a polynomial increase. Since scaling both the matrices
of a bimatrix game by a positive integer does not change the set of Nash equilibria, we assume that entries of $A$, $\gamma$
and $\beta$ are integers.

Now consider the game space $\Gamma=\{(A,-A+\alpha\cdot \beta)\ |\ \alpha \in \mathbb R^m\}$. Clearly, $G(\gamma)=(A,B) \in
\Gamma$ and the corresponding polytopes $P$ and $Q'$ of (\ref{eq3}) are non-degenerate. 
Let $\mathcal N$ be the set of fully-labeled points of $P\times Q'$ as defined in Section \ref{games_polytopes}. 
By Lemma \ref{le1}, we know that for every Nash equilibrium of the game $G(\gamma)$, there is a unique point in $\mathcal
N$. 

Consider the hyper-plane $H: \lambda-\sum_{i=1}^m \gamma_i x_i=0$ in $(y,\pi_1,x,\lambda,\pi_2)$-space and the corresponding
half spaces $H^+: \lambda-\sum_{i=1}^m \gamma_i x_i\geq 0$ and $H^-: \lambda-\sum_{i=1}^m \gamma_i x_i\leq 0$.  
It is easy to see that a point $w\in\mathcal N$ corresponds to a Nash equilibrium of $G(\gamma)$ only if $w\in H$. Therefore
the intersection of $\mathcal N$ with the hyper-plane $H$ gives all the Nash equilibria of $G(\gamma)$. If the hyper-plane
$H$ intersects an edge of $\mathcal N$,
then it intersects the edge exactly at one point, because $G(\gamma)$ is a non-degenerate game. 

Let $\gamma_{min}=\min_{i\in S_1} \gamma_i$ and $\gamma_{max}=\max_{i\in S_1} \gamma_i$. Since $\forall x \in \Delta_1,\
\gamma_{min}\leq\sum_{i=1}^m \gamma_i x_i\leq\gamma_{max}$, a point $w\in \mathcal N$ corresponds to a Nash
equilibrium of $G(\gamma)$, only if the value of $\lambda$ at $w$ is between $\gamma_{min}$ and $\gamma_{max}$. 
From Proposition \ref{ho_pr1}, we know that $\mathcal N$ contains only a path. If we consider the path $\mathcal N$ from the
first edge $(v_s,\mathcal E_{v_s})$ to the last edge $(v_e,\mathcal E_{v_e})$, then $\lambda$ monotonically increases from
$-\infty$ to $\infty$ on it (Lemmas \ref{sle4} and \ref{ho_le3}). Therefore all the points, corresponding to the Nash
equilibrium of $G(\gamma)$ on the path $\mathcal N$, lie between $OPT(\gamma_{min})$ and $OPT(\gamma_{max})$ (Lemma
\ref{ho_le2}).

\subsection{Rank-1 NE: A Polynomial Time Algorithm}\label{binsearch}
Recall that finding a Nash equilibrium of the game $G(\gamma)$ is equivalent to finding a point in the intersection of
$\mathcal N$ and the hyper-plane $H$. As $\lambda$ increases monotonically on $\mathcal N$, and all the points in the
intersection are between the points of $\mathcal N$ corresponding to $\lambda=\gamma_{min}$ and $\lambda=\gamma_{max}$,
the {\em BinSearch} algorithm of Table \ref{a1} applies binary search on $\lambda$ to locate a point in the intersection. 

\begin{table}[!hbt]
\begin{center}
\begin{tabular}{|l|}\hline
{\bf BinSearch}($\gamma_{min},\gamma_{max}$)\\
\hspace{15pt}$a_1\leftarrow\gamma_{min}$; $a_2\leftarrow\gamma_{max}$;\\
\hspace{15pt}{\bf if} IsNE($a_1$) $=0$ {\bf or} IsNE($a_2$) $=0$ {\bf then} {\bf return};\\ 
\hspace{15pt}{\bf while} {\em true}\\
\hspace{30pt}$a\leftarrow \frac{a_1+a_2}{2}$; flag $\leftarrow$ IsNE($a$);\\
\hspace{30pt}{\bf if} flag $=0$ {\bf then} {\bf break};\\
\hspace{30pt}{\bf else if} flag $<0$ {\bf then} $a_1\leftarrow a$; \\
\hspace{30pt}{\bf else} $a_2\leftarrow a$;\\ 
\hspace{15pt}{\bf endwhile} \\
\hspace{15pt}{\bf return}; \\
\\
{\bf IsNE}($\delta$)\\
\hspace{15pt}Find $OPT(\delta)$ by solving $LP(\delta)$;\\
\hspace{15pt}$\overline{u,v}\leftarrow$ The edge containing $OPT(\delta)$; 
$\mathcal H\leftarrow\{w \in \overline{u,v}\ |\ w\in H\}$; \\
\hspace{15pt}{\bf if} $\mathcal H\neq \emptyset$ {\bf then} Output $\mathcal H$; {\bf return} $0$; \\
\hspace{15pt}{\bf else if} $\overline{u,v} \in H^+$ {\bf then return} $1$; \\
\hspace{15pt}{\bf else return} $-1$;\\ 
\hline
\end{tabular}
\end{center}
\caption{BinSearch Algorithm}\label{a1}
\vspace{-0.2cm}
\end{table}

The {\em IsNE} procedure of Table \ref{a1} takes a $\delta \in \mathbb R$ as the input, and outputs a NESP if possible,
otherwise it indicates the position of $OPT(\delta)$ with respect to the hyper-plane $H$. First it finds the optimal set
$OPT(\delta)$ of $LP(\delta)$ and the corresponding edge $\overline{u,v}$ containing $OPT(\delta)$. Next, it finds a set
$\mathcal H$, which consists of all the points in the intersection of $\overline{u,v}$ and the hyper-plane $H$ if any, {\em
i.e.}, Nash equilibria of $G(\gamma)$. Since the game $G(\gamma)$ is non-degenerate, $\mathcal H$ is either a singleton or 
empty. In the former case, the procedure outputs $\mathcal H$ and returns $0$ indicating that a Nash equilibrium has been
found. However in the latter case, it returns $1$ if $\overline{u,v} \in H^+$ otherwise it returns $-1$, indicating the 
position of $\overline{u,v}$ w.r.t. the hyper-plane $H$. 

The {\em BinSearch} algorithm maintains two pivot values $a_1$ and $a_2$ of $\lambda$ such that the corresponding
$OPT(a_1)\in H^-$ and $OPT(a_2)\in H^+$, {\em i.e.}, always on the opposite sides of the hyper-plane $H$.
Clearly $\mathcal N$ crosses $H$ at least once between $OPT(a_1)$ and $OPT(a_2)$. Since $OPT(\gamma_{min})\in H^-$ and
$OPT(\gamma_{max})\in H^+$, the pivots $a_1$ and $a_2$ are initialized to $\gamma_{min}$ and $\gamma_{max}$ respectively.
Initially it calls {\em IsNE} for both $a_1$ and $a_2$ separately and terminates if either returns zero indicating that a
NESP has been found. Otherwise the algorithm repeats the following steps until {\em IsNE} returns zero: It calls {\em IsNE}
for the mid point $a$ of $a_1$ and $a_2$ and terminates if it returns zero. If {\em IsNE} returns a negative value, then
$OPT(a)\in H^-$ implying that $OPT(a)$ and $OPT(a_2)$ are on the opposite sides of $H$, and hence the lower pivot $a_1$
is reset to $a$. In the other case $OPT(a)\in H^+$, the upper pivot $a_2$ is set to $a$, as $OPT(a_1)$ and $OPT(a)$ are on
the opposite sides of $H$. 

Note that, the index of the Nash equilibrium obtained by {\em BinSearch} algorithm is always $+1$,
since $a_1<a_2$ is an invariant (Proposition \ref{pr2}).
For $X\in R^{mn}$, let $\tilde{X}=\max_{_{i\in S_1, j\in S_2}}|x_{ij}|$. Since the column-player's payoff matrix is
represented by $-A+\gamma\cdot \beta^T$ of the game $G(\gamma)$, let $|B|=\max\{\tilde{A},\tilde{\beta},\tilde{\gamma}\}$.

\begin{theorem}
Let $\mathcal L$ be the bit length of the input. The BinSearch terminates in time poly($\mathcal L,m,n$).
\end{theorem}
\begin{proof}
From the above discussion it is clear that the algorithm terminates when the call {\em IsNE($a$)} outputs a NESP of
$G(\gamma)$. 
Let the range of $\lambda$ for an edge $(v,\mathcal E_v) \in \mathcal N$ be $[\lambda_1\ \lambda_2]$.
Let $\Delta=(m+2)!\ (|B|)^{(m+2)}$.
\begin{claim}
$\lambda_2-\lambda_1\geq\frac{1}{\Delta^2}$.
\end{claim}
\begin{proof}
Note that $\lambda_1$ and $\lambda_2$ correspond to the two vertices of $\mathcal E_v \in Q'$. Since $Q'$ is in a
($m+2$)-dimensional space, hence there are $m+2$ equations tight at every vertex of it. Hence both $\lambda_1$ and
$\lambda_2$ are rational numbers with denominator at most $\Delta$. Therefore $\lambda_2-\lambda_1$ is at least
$\frac{1}{\Delta^2}$. \qed
\end{proof}

From the above claim, it is clear that when $a_2-a_1\leq\frac{1}{\Delta^2}$, $OPT(a_1)$ and $OPT(a_2)$ are either part of the
same edge or adjacent edges. In either case, the algorithm terminates after one more call to {\em IsNE}, because {\em IsNE}
checks if the edge corresponding to $OPT(a)$ contains a Nash equilibrium of $G(\gamma)$.  

Clearly $a_2-a_1=\displaystyle\frac{\gamma_{max}-\gamma_{min}}{2^l}$ after $l$ iterations of the {\em while} loop. Let $k$ be
such that 
\[
\begin{array}{l}
\displaystyle\frac{\gamma_{max}-\gamma_{min}}{2^k}=\displaystyle\frac{1}{\Delta^2} \Rightarrow 2^k=\Delta^2
(\gamma_{max}-\gamma_{min}) \Rightarrow \\ \\ \hspace{2cm}k=2\log\Delta + \log(\gamma_{max}-\gamma_{min})\leq O(m \log m + m \log |B|+
\log(\gamma_{max}-\gamma_{min})) 
\end{array}
\]

{\em BinSearch} makes at most $k+1$ calls to the procedure {\em IsNE}, which is polynomial in 
$\mathcal L,m$, and $n$. The procedure {\em IsNE} solves a linear program and computes a set $\mathcal H$, both may be done
in poly($\mathcal L,m,n$) time. Therefore the total time taken by {\em BinSearch} is polynomial in $\mathcal L,m$, and $n$.
\qed 
\end{proof}

\subsection{Enumeration Algorithm for Rank-1 Games}\label{enum}

The {\em Enumeration} algorithm of Table \ref{a2} simply follows the path $\mathcal N$ between $OPT(\gamma_{min})$ and
$OPT(\gamma_{max})$, and outputs the NESPs whenever it hits the hyper-plane $H: \lambda-\sum_{i=1}^m \gamma_i x_i=0$.

\begin{table}[h]
\begin{center}
\begin{tabular}{|l|}\hline
{\bf Enumeration}($\overline{u_1,v_1}$, $\overline{u_2,v_2}$)\\
\hspace{15pt}$\overline{u,u'}\leftarrow \overline{u_1,v_1}$;\\
\hspace{15pt}{\bf if} $\overline{u,u'}$ of type $(v,\mathcal E_v)$ {\bf then} flag $\leftarrow 1$; \\
\hspace{15pt}{\bf else} flag $\leftarrow 0$; \\
\hspace{15pt}{\bf while} {\em true}\\
\hspace{30pt}$\mathcal H=\{w \in \overline{u,u'}\ |\ w\in H\}$; Output $\mathcal H$;\\
\hspace{30pt}{\bf if}  $\overline{u,u'}= \overline{u_2,v_2}$ {\bf then} break; \\
\hspace{30pt}{\bf if} flag $=1$ {\bf then} $\overline{u,u'}\leftarrow(\mathcal E_{u'},u')$; flag $\leftarrow 0$;\\
\hspace{30pt}{\bf else} $\overline{u,u'}\leftarrow(u',\mathcal E_{u'})$; flag $\leftarrow 1$; \\
\hspace{15pt}{\bf endwhile} \\
\hspace{15pt}{\bf return};\\
\hline
\end{tabular}
\end{center}
\caption{Enumeration Algorithm}\label{a2}
\end{table}

We obtain $OPT(\gamma_{min})$ and $OPT(\gamma_{max})$ on the path $\mathcal N$ by solving $LP(\gamma_{min})$ and
$LP(\gamma_{max})$ respectively. Let the edges $\overline{u_1,v_1}$ and $\overline{u_2,v_2}$ contain $OPT(\gamma_{min})$ and
$OPT(\gamma_{max})$ respectively.  The call {\em Enumeration}($\overline{u_1,v_1}$, $\overline{u_2,v_2}$) enumerates all the
Nash equilibria of the game $G(\gamma)$.

The {\em Enumeration} algorithm initializes $\overline{u,u'}$ to the edge $\overline{u_1,v_1}$. Since the edges alternate
between the type $(v,\mathcal E_v)$ and $(\mathcal E_w,w)$ on $\mathcal N$, the value of {\em flag} indicates the type of
edge to be considered next. It is set to one if the next edge is of type $(\mathcal E_w,w)$, otherwise it is set to zero. In
the while loop, it first outputs the intersection of the edge $\overline{u,u'}$ and the hyper-plane $H$, if any. Further, if
the value of the flag is one then $\overline{u,u'}$ is set to $(\mathcal E_{u'},u')$, otherwise it is set to $(u',\mathcal
E_{u'})$, and the flag is toggled. Recall that the edges incident on a vertex $u'$ in $\mathcal N$ may be obtained by
relaxing the inequality corresponding to the duplicate label of $u'$, in $P$ and in $Q'$ (Section \ref{games_polytopes}). Let
the duplicate label of the vertex $u'$ be $i$. We may obtain the edge $(\mathcal E_{u'},u')$ by relaxing the inequality $i$
of $P$ and the edge $(u',\mathcal E_{u'})$ by relaxing the inequality $i$ of $Q'$. The algorithm terminates when
$\overline{u,u'}=\overline{u_2,v_2}$.

Every iteration of the loop takes time polynomial in $\mathcal L$, $m$ and $n$. Therefore, the time taken by the algorithm is
equivalent to the number of edges on $\mathcal N$ between $\overline{u_1,v_1}$ and $\overline{u_2,v_2}$. 

For a general bimatrix game $(A,B)$, we may obtain $C$, $\gamma$ and $\beta$ such that $B=C+\gamma\cdot\beta^T$, and define
the corresponding game space $\Gamma$ and the polytopes $P$ and $Q'$ accordingly (Section \ref{games_polytopes}). There is a
one-to-one correspondence between the Nash equilibria of the game $(A,B)$ and the points in the intersection of the
fully-labeled set $\mathcal N$ and the hyper-plane $\lambda - \sum_{i=1}^m \gamma_i x_i=0$. Recall that the set $\mathcal N$
contains one path ($\mathcal P$) and a set of cycles (Proposition \ref{pr1}). The extreme edges $(v_s,\mathcal E_{v_s})$ and
$(v_e,\mathcal E_{v_e})$ of $\mathcal P$ may be easily obtained as described in the proof of Lemma \ref{sle4}. Since $\mathcal
P$ contains at least one Nash equilibrium of every game in $\Gamma$ (Lemma \ref{le_cover}), hence the call
$Enumerate((v_s,\mathcal E_{v_s}),(v_e,\mathcal E_{v_e}))$ outputs at least one Nash equilibrium of the game $(A,B)$. Note
that the time taken by the algorithm again depends on the number of edges on the path $\mathcal P$.

\subsubsection{Comparison with Earlier Approaches.}
The {\em Enumeration} algorithm may be compared to two previous algorithms. One is the Theobald algorithm \cite{the}, which
enumerates all Nash equilibria of a rank-$1$ game, and the other is the Lemke-Howson algorithm \cite{lem}, which finds a Nash
equilibrium of any bimatrix game. The {\em Enumeration} algorithm enumerates all the Nash equilibria of a rank-$1$ game and
for any general bimatrix game it is guaranteed to find one Nash equilibrium. All three algorithms are path following
algorithms. However, the main difference is that both the previous algorithms always trace a path on the best response
polytopes of a given game ({\em i.e.,} $P(\gamma)\times Q(\gamma)$), while the {\em Enumeration} algorithm follows a path on a
bigger polytope $P\times Q'$ which encompasses best response polytopes of all the games of an $m$-dimensional game space.
Therefore, for every game in this $m$-dimensional game space, the {\em Enumeration} follows the same path. 
Further, all the points on the path followed by {\em Enumeration} algorithm are
fully-labeled, and it always hits the best response polytope of the given game at one of it's NESP points. However the path
followed by previous two algorithms is not fully-labeled and whenever they hit a fully-labeled point, it is a NESP of the
game. 

In every intermediate step, the Theobald algorithm calculates the range of a variable ($\xi$) based on the feasibility of
primal and dual, and accordingly decides which inequality to relax (in $P$ or $Q$) to locate the next edge. While {\em
Enumeration} algorithm simply leaves the duplicate label in $P$ or $Q'$ (alternately) at the current vertex to 
locate the next edge. Further, for a general bimatrix game, the {\em Enumeration} algorithm locates at least one Nash
equilibrium, while Theobald algorithm works only for rank-$1$ games.

For rank-$1$ games there may be a polynomial bound for the {\em Enumeration} algorithm, because experiments suggest that
the path $\mathcal N$ contains very few edges for randomly generated rank-$1$ games. 

\section{Rank-$k$ Space and Homeomorphism}\label{fixedRank}

It turns out that the approach used to show the homeomorphism between the subspace of rank-$1$ games and it's Nash
equilibrium correspondence may be extended to the subspace with rank-$k$ games.  Given a bimatrix game $(A,B)\in R^{2mn}$ of
rank-$k$, the matrix $A+B$ may be written as $\sum_{l=1}^k \gamma^l\cdot\beta^{l^T}$, using the linearly independent vectors
$\gamma^l \in \mathbb R^m,\ \beta^l \in \mathbb R^n, 1\leq l\leq k$.  Therefore, the column-player's payoff matrix $B$ may be
written as $B=-A+\sum_{l=1}^k \gamma^l\cdot\beta^{l^T}$.  Consider the corresponding game space
$\Gamma^k=\{(A,-A+\sum_{l=1}^k \alpha^l\cdot\beta^{l^T})\in \mathbb R^{2mn}\ |\ \forall l\leq k,\ \alpha^l \in \mathbb
R^m\}$, where $\{\beta^l\}_{l=1}^k$ are linearly independent. This space is an affine $km$-dimensional subspace of the
bimatrix game space $\mathbb R^{2mn}$, and it contains only rank-$k$ games. Let $\alpha=(\alpha^1,\dots,\alpha^k)$, and
$G(\alpha)$ denote the game $(A,-A+\sum_{l=1}^k \alpha^l\cdot\beta^{l^T})$. The Nash equilibrium correspondence of the space
$\Gamma^k$ is $E_{\Gamma^k}=\{(\alpha,x,y)\in \mathbb R^{km}\times\Delta_1\times\Delta_2\ |\ (x,y) \mbox{ is a NESP of }
G(\alpha)\in \Gamma^k\}$. 

For all the games in $\Gamma^k$, again the row-player's payoff matrix remains constant, hence for all $G(\alpha)\in \Gamma^k$
the best response polytope of the row-player $P(\alpha)$ is $P$ of (\ref{eq2}).  However, the best response polytope of the
column player $Q(\alpha)$ varies, as the payoff matrix of the column-player varies with $\alpha$. Consider the following
polytope (similar to (\ref{eq3})). 
\begin{eqnarray}\label{eq_Qnew}
Q'^k=\{(x,\lambda,\pi_2)\in \mathbb R^{m+k+1}\ |\ x_i\geq 0,\ \forall i \in S_1;\hspace{6.5cm}\\
x^T(-A^j)+\sum_{l=1}^k\beta^l_j\lambda_l-\pi_2\leq 0,\ \forall j \in S_2;\ \sum_{i=1}^m x_i=1\}\hspace{0.5cm}\nonumber
\end{eqnarray}

Note that $\lambda=(\lambda_1,\dots,\lambda_k)$ is a variable vector.
The column-player's best response polytope $Q(\alpha)$, for the game $G(\alpha)$, is the projection of the set
$\{(x,\lambda,\pi_2)\in Q'^k\ |\ \forall l\leq k,\ \sum_{i=1}^m \alpha^l_i x_i-\lambda_l = 0\}$ on $(x,\pi_2)$-space. We
assume that the polytopes $P$ and $Q'^k$ are non-degenerate. Let the set of fully-labeled pairs of
$P\times Q'^k$ be $\mathcal N^k=\{(v,w)\in P\times Q'^k\ |\ L(v)\cup L(w)=\{1,\dots,m+n\}\}$. The following facts regarding
the set $\mathcal N^k$ may be easily derived.

\begin{itemize}
\item For every point in $E_{\Gamma^k}$ there is a unique point in $\mathcal N^k$, and for every point in $\mathcal N^k$
there is a point in $E_{\Gamma^k}$ (Lemma \ref{le1}). Further the set of points of $E_{\Gamma^k}$ mapping to a point
$(v,w)\in \mathcal N^k$, is equivalent to $k(m-1)$-dimensional space. 
\item Since there are $k$ more variables in $Q'^k$, namely $\lambda_1,\dots,\lambda_k$ compared to $Q$ of (\ref{eq2}), 
$\mathcal N^k$ is a subset of the $k$-skeleton of $P\times Q'^k$. If a point $v \in P$ is on a $d$-dimensional face ($d\le
k$), then the set $\mathcal E_v$ is either empty or it is a $(k-d)$-dimensional face, where $\mathcal E_v=\{w\in Q'^k\ |\
(v,w)\in \mathcal N^k\}$ (Observations of Section \ref{games_polytopes}).  
\item For every $(v,w)=((y,\pi_1),(x,\lambda,\pi_2))$ in $P\times Q'^k$, $\sum_{l=1}^k \lambda_l (\beta^{l^T} \cdot
y)-\pi_1-\pi_2\leq 0$, and equality holds iff $(v,w)\in \mathcal N^k$. 
\end{itemize}

For a vector $\delta \in \mathbb R^k$, consider the following parametrized linear program $LP^k(\delta)$.
\begin{eqnarray}\label{eq_lpnew}
\begin{array}{ll}
\vspace{5pt}
LP^k(\delta):\hspace{.1in}  & \max\ \sum_{l=1}^k \delta_l (\beta^{l^T} \cdot y)-\pi_1-\pi_2 \\
&\hspace{.4in} (y,\pi_1)\in P \\
&\hspace{.4in} (x,\lambda,\pi_2) \in Q'^k \\
&\hspace{.4in} \lambda_l=\delta_l, \forall l\leq k
\end{array}
\end{eqnarray}

Let $OPT^k(\delta)$ be the set of optimal points of $LP^k(\delta)$. Note that for any $a\in \mathbb R^k$, all the points on
$\mathcal N^k$ with $\lambda=a$ may be obtained by solving $LP^k(a)$. In other words, $\{((y,\pi_1),(x,\lambda,\pi_2))\in
\mathcal N^k\ |\ \lambda=a\}=OPT^k(a)$ (Lemma \ref{ho_le2}).  Using this fact we show that the tuple
$(\lambda_1+\beta^{1^T}\cdot y,\dots,\lambda_k+\beta^{k^T}\cdot y)$ uniquely identifies a point of $\mathcal N^k$. For a
vector $a \in \mathbb R^k$, let $S(a)=\{((y,\pi_1),(x,\lambda,\pi_2))\in \mathcal N^k\ |\ \forall l\leq k,\
\lambda_l+\beta^{l^T}\cdot y=a_l\}$. 
\\

\noindent{\bf Lemma A}.\label{dis_le1}
{\em For a vector $a \in \mathbb R^k$, the set $S(a)$ contains exactly one element, i.e., $|S(a)|=1$.}
\begin{proof}
First we show that $S(a)\neq \emptyset$. 
Let $S_1(a)=\{((y,\pi_1),(x,\lambda,\pi_2))\in \mathcal N^k\ |\ \forall l>1,\ \lambda_l+\beta^{l^T}\cdot y=a_l\}$. 
Using the similar analysis as in Lemmas \ref{sle4} and \ref{le_lambda}, it may be easily shown that for every $b\in \mathbb R$
there is a point in $S_1(a)$ such that $\lambda_1+\beta^{1^T}\cdot y=b$. Therefore $S(a)\neq \emptyset$.

Now, suppose $|S(a)|>1$ implying that there are at least two points $(v_1,w_1)$ and $(v_2,w_2)$ in $S(a)$. Let
$v_i=(y^i,\pi_1^i)$, $w_1=(x^1,c,\pi_2^1)$ and $w_2=(x^2,d,\pi_2^2)$. Clearly, $(v_1,w_1)$ and $(v_2,w_1)$ are feasible
points of $LP^k(c)$ and $(v_1,w_1)\in OPT^k(c)$. Similarly, $(v_2,w_2)$ and $(v_1,w_2)$ are feasible points of $LP^k(d)$ and
$(v_2,w_2)\in OPT^k(d)$. Therefore the following holds.
\[
\begin{array}{l}
\sum_{l=1}^k c_l (\beta^{l^T}\cdot y^1) - \pi_1^1-\pi_2^1\geq \sum_{l=1}^k c_l (\beta^{l^T}\cdot y^2) - \pi_1^2-\pi_2^1 \\
\sum_{l=1}^k d_l (\beta^{l^T}\cdot y^2) - \pi_1^2-\pi_2^2\geq \sum_{l=1}^k d_l (\beta^{l^T}\cdot y^1) - \pi_1^1-\pi_2^2 \\
\end{array}
\]
Using the fact that $\beta^{l^T}\cdot y=a_l-\lambda_l,\ \forall l\leq k$ and the above equations, we get
\[
\begin{array}{r}
\sum_{l=1}^k c_l(a_l-c_l) + d_l(a_l-d_l) \geq \sum_{l=1}^k c_l(a_l-d_l) + d_l(a_l-c_l)\\
\begin{array}{l}
\Rightarrow 
-\sum_{l=1}^k (c_l-d_l)^2 \geq 0 \\
\Rightarrow
\forall l\leq k,\ c_l=d_l \Rightarrow c=d\\
\Rightarrow \forall l\leq k,\ \beta^{l^T}\cdot y^1=\beta^{l^T}\cdot y^2 
\end{array}
\end{array}
\]
The above expressions and the fact that $\sum_{l=1}^k\lambda_l(\beta^{l^T}\cdot y)-\pi_1-\pi_2$ evaluates to zero at both
$(v_1,w_1)$ and $(v_2,w_2)$ imply that $\pi_1^1=\pi_1^2$ and $\pi_2^1=\pi_2^2$. Note that, $S(a)\subset OPT^k(c)$. 
\begin{claim}
The set $\{w\in Q'^k \ |\ (v,w) \in OPT^k(c),\ v\in P\}$ is a singleton.
\end{claim}
\begin{proof}
Suppose the set $\{w\in Q'^k \ |\ (v,w) \in OPT^k(c),\ v\in P\}$ contains two distinct points $w$ and $w'$.
In that case, $\lambda$ takes value $c$ on the $1$-dimensional line $L\subset Q'^k$ containing both $w$ and $w'$.
Note that the points corresponding to the end-points of $L$ are on the lower dimensional face ($<k$) of $Q'^k$ and both these
points make separate convex sets of fully labeled pairs with the points of $P$. Further the convex hull of these two convex
sets is not contained by $\mathcal N^k$, however both these sets are contained in $OPT^k(c)$ and $OPT^k(c)\subset N^k$. It 
implies that $OPT^k(c)$ is not convex, which is a contradiction. \qed
\end{proof}

The above claim implies that $w=w_1=w_2$. Now it is enough to show that $v_1=v_2$ to prove the lemma. In the extreme case,
$w$ is a vertex of $Q'^k$ and makes a fully-labeled pair with a $k$-dimensional face of $P$.  Let
$M(w)=\{1,\dots,m+n\}\setminus
L(w)$. Clearly, $|M(w)|\geq n-k$, $M(w)\subseteq L(v_1)$ and $M(w)\subseteq L(v_2)$. Suppose $v_1\neq v_2$, then on the line
joining $v_1$ and $v_2$, the following equations are tight: $\beta^l\cdot y=a_l-c_l,\ \forall l\leq k$;
$\sum_{j=1}^n y_j=1$ and all the equations corresponding to $M(w)$. Clearly, there are at least $n+1$ equations tight on
this line and they are not linearly independent. This contradicts the fact that $A$ and $\beta^l$s are generic.
\qed
\end{proof}

Motivated by Lemma \hyperref[dis_le1]{A}, we consider the function $g^k:\mathcal N^k \rightarrow \mathbb R^k$ such that, 
\begin{eqnarray}\label{eq_gk}
g^k((y,\pi_1),(x,\lambda,\pi_2))=(\lambda_1+(\beta^{1^T}\cdot y),\dots,\lambda_k+(\beta^{k^T}\cdot y))
\end{eqnarray}
The function $g^k$ is continuous and bijective (Lemma \hyperref[dis_le1]{A}), and the inverse $g^{k^{-1}}:\mathbb R^k\rightarrow
\mathcal N^k$ is also continuous, since $\mathcal N^k$ is a closed and connected set. Using $g^{k}$ and a function similar to
(\ref{eq4}), we establish the homeomorphism between $\Gamma^k$ and $E_{\Gamma^k}$. 

\begin{theorem}\label{dis_th}
The Nash equilibrium correspondence $E_{\Gamma^k}$ is homeomorphic to the game space $\Gamma^k$.
\end{theorem}
\begin{proof}
Consider the function $f^k:E_{\Gamma^k}\rightarrow \Gamma^k$ as follows:
\[
f^k(\alpha,x,y)=(\alpha'^1,\dots,\alpha'^k), \mbox{ where } \alpha'^l=(\lambda_l+(\beta^{l^T}\cdot
y),\ \alpha^l_2-\alpha^l_1,\dots,\ \alpha^l_m-\alpha^l_1)^T,\ \forall l\leq k
\]
\begin{claim}
Function $f^k$ is bijective.
\end{claim}
\begin{proof}
Consider an $\alpha'=(\alpha'^1,\dots,\alpha'^k) \in \mathbb R^{mk}$. We construct a point $(\alpha,x,y)\in E_{\Gamma^k}$
such that $f^k(\alpha,x,y)=\alpha'$. Let $((y,\pi_1),(x,\lambda,\pi_2))=g^{k^{-1}}(\alpha'^1_1,\dots,\alpha'^k_1)$. Now we
solve the following system of equations to get $\alpha$.
\[
\begin{array}{rl}
\vspace{.15cm}
\forall l\leq k,&\hspace{10pt} \displaystyle\sum_{i=1}^m x_i \alpha_i^l = \lambda_l\\
\forall l\leq k,\ \forall i>1,& \hspace{10pt}\alpha_i^l = \alpha'^l_i-\alpha_1^l
\end{array}
\]
It is easy to see that we get a unique $\alpha$ by solving the above equations, and $f^k(\alpha,x,y)=\alpha'$ holds. \qed
\end{proof}
From the claim, it is clear that $f^k$ is a continuous bijective function. The inverse function $f^{k^{-1}}:\Gamma^k
\rightarrow E_{\Gamma^k}$ is also continuous, since $g^{k^{-1}}$ is continuous and the set $E_{\Gamma^k}$ is closed and
connected.  \qed
\end{proof}

Using the above theorem, next we give a fixed point formulation to solve a rank-$k$ game.

\begin{lemma}
Finding a Nash equilibrium of a game $G(\gamma) \in \Gamma^k$ reduces to finding a fixed point of a polynomially 
computable piece-wise linear function $f:[0,1]^k\rightarrow[0,1]^k$.
\end{lemma}
\begin{proof}
Consider the hyper-planes $H_l:\lambda_l-\sum_{i=1}^m \gamma_i^l x_i=0,\ \forall l\leq k$ and the
corresponding half spaces $H_l^-: \lambda_l-\sum_{i=1}^m \gamma_i^l x_i\leq 0,\ H_l^+: \lambda_l-\sum_{i=1}^m \gamma_i^l
x_i\geq 0,\ \forall l\leq k$.
A point $u\in \mathcal N^k$ corresponds to a NESP of the game $G(\gamma)$, iff $\forall l\leq k, u \in H_l$.
We know that for any $a \in \mathbb R^k$, the points on $\mathcal N^k$ with $\lambda=a$ are the optimal points of $LP^k(a)$,
{\em i.e.}, $OPT^k(a)$. 

Let $\gamma_{min}=(\gamma^1_{min},\dots,$ $\gamma^k_{min})$, where $\gamma^l_{min}=\min_{i\in S_1} \gamma^l_i,\ \forall l\leq k$
and $\gamma_{max}=(\gamma^1_{max},\dots,\gamma^k_{max})$, where $\gamma^l_{max}=\max_{i\in S_1} \gamma^l_i,$ $ \forall l\leq k$.
Consider the box $\mathcal B \in \mathbb R^k$ such that $\mathcal B=\{a\in \mathbb R^k\ |\ \gamma_{min}\leq a\leq
\gamma_{max}\}$\footnote{For any two vectors $a,b\in \mathbb R^n$, by $a\leq b$ we mean $a_i\leq b_i,\ \forall i\leq n$.}.
For the rank-$1$ case, $\mathcal B$ is an interval. We may obtain $OPT^k(\gamma_{min})$ and $OPT^k(\gamma_{max})$ by solving
$LP^k(\gamma_{min})$ and $LP^k(\gamma_{max})$ respectively. Clearly, $OPT^k(\gamma_{min})\in \bigcap_{l\leq k} H_l^-$ and 
$OPT^k(\gamma_{max})\in \bigcap_{l\leq k} H_l^+$. 
It is easy to see that, all the $a \in \mathbb R^k$ such that $\mathcal N^k$ intersects all the hyper-planes ($H_l$) together
at $OPT^k(a)$, lies in the box $\mathcal B$.

The points corresponding to the Nash equilibria of the game $G(\gamma)$ may also be modeled as the fixed points of the
function $f:\mathcal B\rightarrow \mathcal B$ such that,
\[
f(a)=(\sum_{i=1}^m \gamma^1_i x_i,\ \dots,\ \sum_{i=1}^m \gamma^k_i x_i), 
\mbox{ where } (x,\lambda,\pi_2)=\{w\in Q'^k\ |\ (v,w)\in OPT^k(a),\ v \in P\}
\]
For every $a \in \mathcal B$, the corresponding $x$ is well defined in the above expression (Proof of Lemma
\hyperref[dis_le1]{A}),
and may be obtained in polynomial time by solving $LP^k(a)$.
It is easy to see that the function $f$ is a piece-wise linear function. 
\qed
\end{proof}

It seems that for a given $a\in \mathbb R^k$, there is a way to trace the points in the intersection of $\mathcal N^k$ and
$\lambda_l=a_l, l\neq i$, such that $\lambda_i$ increases monotonically (analysis similar to Lemma \ref{ho_le3}). Using this and
the simple structure of $\mathcal N^k$, is there a way to locate a fixed point of $f$ in polynomial time?

\section{Conclusion}\label{conc}
In this paper, we establish a homeomorphism between an $m$-dimensional affine subspace $\Gamma$ of the bimatrix game space
and it's Nash equilibrium correspondence $E_\Gamma$, where $\Gamma$ contains only rank-$1$ games. To the best of our
knowledge, this is the first structural result for a subspace of the bimatrix game space.  The homeomorphism maps that we
derive are very different than the ones given by Kohlberg and Mertens for the bimatrix game space \cite{koh} and it builds on
the structure of $E_\Gamma$.  Further, using this structural result we design two algorithms. The first algorithm finds a
Nash equilibrium of a rank-$1$ game in polynomial time. This settles an open question posed by Kannan and Theobald
\cite{kan} and Theobald \cite{the}. The second algorithm enumerates all the Nash equilibria of a rank-$1$ game and finds at
least one Nash equilibrium of a general bimatrix game.

Further, we extend the above structural result by establishing a homeomorphism between $km$-dimensional affine subspace
$\Gamma^k$ and it's Nash equilibrium correspondence $E_{\Gamma^k}$, where $\Gamma^k$ contains only rank-$k$ games. We hope
that this homeomorphism result will help in designing a polynomial time algorithm to find a Nash equilibrium of a fixed rank
game.

\begin{appendix}
\section{Regions in the Game Space}\label{region}
In this section, we analyze the structure of $E_\Gamma$ in detail. For every vertex $v \in \mathcal N^P$, first we identify a
region in the game space and the points in $E_\Gamma$ corresponding to the region. Later we combine them to get the complete
structure of $E_\Gamma$. 

For a vertex $v=(y,\pi_1)$ of $P$, let $R(v)=\{\alpha\ |\ (\alpha,x,y) \in E_\Gamma, \mbox{ for some } (x,\lambda,\pi_2)\in
\mathcal E_v\}$ be it's region in the game space, {\em i.e.}, the set of games with at least one NE corresponding to $v$.
Clearly, $R(v)$ is non-empty only when $v \in \mathcal N^P$. For a $w=(x,\lambda,\pi_2) \in \mathcal E_v$, let $H_w$ be the
hyper-plane $\sum_{i=1}^m x_i \alpha_i - \lambda=0$ in the game space. By $\alpha' \in H_w$ we mean $H_w(\alpha')=0$.  Recall
that for a game $G(\alpha) \in \Gamma$ the row-player's best response polytope is $P(\alpha)=P$, and column-player's
best response polytope is $Q(\alpha)$ which may be obtained by replacing $\lambda$ by $\sum_{i=1}^m \alpha_i x_i$ in $Q'$ of 
(\ref{eq3}). 

\begin{lemma}\label{le51}
Let $v=(y,\pi_1) \in \mathcal N^P$ be a vertex. 
$\alpha' \in R(v)$ iff $\exists w \in \mathcal E_v$ such that $H_w(\alpha')=0$. 
\end{lemma}
\begin{proof}
($\Rightarrow$) Suppose $\alpha' \in R(v) \Rightarrow (\alpha',x',y) \in E_\Gamma$ for some $u'=(x',\pi_2) \in Q(\alpha')$.
$(\alpha',x',y) \in E_\Gamma \Rightarrow (v,u')$ makes a fully-labeled pair of $P(\alpha')\times Q(\alpha')$. Let
$w'=(x',\lambda',\pi_2)$, where $\lambda'=\sum_{i=1}^m \alpha'_i x'_i$. Clearly, $w' \in Q'$ and $L(u')=L(w')\Rightarrow 
w' \in \mathcal E_v$, and $H_{w'}(\alpha')=0$.

($\Leftarrow$)
For a $w=(x,\lambda,\pi_2)\in \mathcal E_v$ consider a point $\alpha'\in H_w$. Clearly, the $Q(\alpha')$ of the
game $G(\alpha')$ has a vertex $(x,\pi_2)$ with the same set of tight equations as $w$, and it makes fully-labeled vertex
pair with $v$. This makes $(x,y)$ a NESP of $G(\alpha')\Rightarrow (\alpha',x,y) \in E_\Gamma \Rightarrow \alpha' \in R(v)$. 
\qed
\end{proof}

Lemma \ref{le51} implies that $R(v)=\displaystyle\bigcup_{\forall w \in \mathcal E_v} H_w$. The following lemmas identify the
boundary of $R(v)$. 

\begin{lemma}\label{le6}
Let $(I,J)$ be the support-pair of $v \in P$ with $\mathcal E_v\neq \emptyset$. 
\begin{enumerate}
\item If $|I|=|J|\geq 2$, then $R(v)$ is a union of two convex sets, and it is defined by only two hyper-planes.  
\item If $|I|=|J|=1$, then $R(v)$ is a convex-set. It has one defining hyper-plane if $v$ is either $v_s$ or $v_e$, otherwise
it has two parallel defining hyper-planes.  
\end{enumerate}
\end{lemma}
\begin{proof}
For the first part, let the bounding vertices of edge $\mathcal E_v \in \mathcal N^{Q'}$ be $w_1$ and $w_2$ (Lemma
\ref{sle4}). Every point $w \in \mathcal E_v$ may be written as a convex combination of $w_1$ and $w_2$. Therefore, the
corresponding hyper-plane $H_w$ may be written as a convex combination of the hyper-planes $H_{w_1}$ and $H_{w_2}$. This
implies that $\forall w \in \mathcal E_v,\ H_{w_1}\cap H_{w_2} \subset H_w$. Further, it is easy to see that the union of
convex sets $\{H_{w_1}\geq 0,\ H_{w_2}\leq 0\}$ and $\{H_{w_1}\leq 0,\ H_{w_2}\geq 0\}$ forms the region $R(v)$ (Lemma
\ref{le51}), and hence the hyper-plane $H_{w_1}$ and $H_{w_2}$ defines the boundary of $R(v)$. 

For the second part if $v=v_s$ or $v=v_e$ then the corresponding edge $\mathcal E_v \in \mathcal N^{Q'}$ has exactly
one vertex $w_1$ (Lemma \ref{sle4}), and hence there is exactly one defining hyper-plane of $R(v)$, namely $H_{w_1}$.
Moreover, for $v=v_s$ and $v=v_s$ the region $R(v)$ is defined by $H_{w_1}\leq 0$ and $H_{w_1}\geq 0$ respectively. 

If $v\neq v_s$ and $v\neq v_e$, then the edge $\mathcal E_v$ has two bounding vertices $w_1$ and $w_2$, and $x$ remains
constant on $\mathcal E_v$ (Lemma \ref{sle4}). Therefore, the hyper-planes $H_{w_1}$ and $H_{w_2}$ are parallel to each other.
Further, since any point $w \in \mathcal E_v$ may be written as a convex combination of $w_1$ and $w_2$, the hyper-plane
$H_w$ lies between $H_{w_1}$ and $H_{w_2}$. Hence, the hyper-planes $H_{w_1}$ and $H_{w_2}$ define the boundary of the region
$R(v)$ (Lemma \ref{le51}).  \qed
\end{proof}

Lemma \ref{le6} shows that the regions are very simple and they are defined by at most two hyper-planes. Moreover, if
$H_{w_1}$ and $H_{w_2}$ are the two defining hyper-planes of $R(v)$ then $\forall w \in \mathcal E_v,\ H_{w_1}\cap H_{w_2}
\subset H_w$.  Next, we discuss how the adjacency of $v$ in $\mathcal N^P$ carries forward to the adjacency of the regions
through the corresponding defining hyper-planes.

\begin{lemma}\label{le8}
If the edges $\mathcal E_v$ and $\mathcal E_{v'}$ share a common vertex $w \in \mathcal N^{Q'}$, then the hyper-plane $H_w=0$
forms a boundary of both $R(v)$ and $R(v')$.
\end{lemma}
\begin{proof}
Lemma \ref{le6} establishes a one-to-one correspondence between the bounding vertices of $\mathcal E_v$ and the defining
hyper-planes of the region $R(v)$. For every bounding vertex $w$ of $\mathcal E_v$, there is a defining hyper-plane $H_w$ of
$R(v)$ and vice-versa, and $H_w\subset R(v)$. 
\qed
\end{proof}

Clearly, $R(v)$ and $R(v')$ are adjacent through a common defining hyper-plane $H_w$, where $w=\mathcal E_v\cap
\mathcal E_{v'}$ is a vertex. Moreover, for every defining hyper-plane of $R(v)$ there is a unique adjacent region. Hence, every region
has at most two adjacent regions and there are exactly two regions with only one adjacent region (Lemmas \ref{le6} and
\ref{le8}). In short adjacency of vertices of $\mathcal N^P$ carries forward to the regions.

Let {\em region graph} be the graph, where for every non-empty region $R(v)$ there is a node in the graph, and two nodes are
connected iff the corresponding regions are adjacent.  Clearly, the degree of every node in this graph is at most two and
there are exactly two nodes with degree one. The region graph consists of a path and a set of cycles, and it is isomorphic to
$\mathcal N^P$ where a vertex $v\in \mathcal N^P$ is mapped to the vertex $R(v)$. Therefore, for every component of $\mathcal
N$, we get a component of the region graph. 

To identify a component of the region graph with a component of $E_\Gamma$, first we distinguish the part of $E_\Gamma$
related to $R(v)$.  For an $\alpha \in R(v)$, let $\mathcal S(\alpha)=E_\Gamma \cap \{(\alpha,x,y)\ |\ y\in \Delta_2 \mbox{
and } (x,\lambda,\pi_2)\in \mathcal E_v\}$. Let $v=(y,\pi_1)\in \mathcal N^P$ be a vertex and $\overline{w_1,w_2}=\mathcal
E_v\in \mathcal N^{Q'}$. Let $(x^1,\lambda_1,\pi_2^1)=w_1$, $(x^2,\lambda_2,\pi_2^2)=w_2$, $\overline{v_1,v}=\mathcal
E_{w_1}\in \mathcal N^P$ and $\overline{v,v_2}=\mathcal E_{w_2}\in \mathcal N^P$. We may easily deduce the following facts. 

\begin{enumerate}
\item Let $w'=(x',\lambda',\pi'_2) \in \mathcal E_v$ be a non-vertex point and $\alpha' \in H_{w'}\setminus (H_{w_1}\cap
H_{w_2})$, then $\mathcal S(\alpha')=\{(\alpha',x',y)\}$.
\item For $\alpha' \in H_{w_1}\setminus H_{w_2}$, $\mathcal S(\alpha')=\{(\alpha',x^1,y')\ |\ (y',\pi'_1)\in
\overline{v_1,v}\}$. Similarly, for $\alpha' \in H_{w_2}\setminus H_{w_1}$, $\mathcal S(\alpha')=\{(\alpha',x^2,y')\ |\
(y',\pi'_1)\in\overline{v,v_2}\}$.  
\item For $\alpha' \in H_{w_1}\cap H_{w_2}$, $\mathcal S(\alpha')=\{(\alpha',x',y')\ |\ ((y',\pi'_1),(x',\lambda',\pi'_2))\in
(\overline{v_1,v},w_1)\cup (v,\overline{w_1,w_2})\cup (\overline{v,v_2},w_2)\}$. From Lemma \ref{le6}, $\forall w\in\mathcal
E_v,\ H_{w_1} \cap H_{w_2} \subset H_w$. Therefore the projection of edges $(\overline{v_1,v},w_1),\ (v,\overline{w_1,w_2}),\
(\overline{v,v_2},w_2)$ on $(y,\pi_1,x,\pi_2)$-space is contained in $P(\alpha')\times Q(\alpha')$ and all the points on them
are fully-labeled. 
\end{enumerate}

The above facts imply that $(x,y)$ changes continuously inside the region ($R(v)$) as well as on the boundary
($H_{w_1},H_{w_2}$), and their values come from the corresponding adjacent edges of $\mathcal N$ ($(\overline{v_1,v},w_1)$,
$(v,\overline{w_1,w_2}),\ (\overline{v,v_2},w_2)$). Moreover, the consistency is maintained across the regions through the
NESPs of the games on the common defining hyper-plane. 

All of these imply that there is a path between two points of $E_\Gamma$ iff the corresponding points in $\mathcal N$ lie on
the same component of $\mathcal N$. Therefore $E_\Gamma$ does not form a single connected component if $\mathcal N$
has more than one component. From the discussion in Section \ref{games_polytopes}, we know that $\mathcal N$ contains at
least a path and may contain some cycles. Hence $E_\Gamma$ forms a single connected component iff $\mathcal N$ contains only
the path. Example \ref{ex1} illustrates that $E_\Gamma$ is not connected in general by illustrating a $\mathcal N$ with
cycles.
\end{appendix}

\begin{thebibliography}{99}
\bibitem{bulow} Bulow, J., Levin, J.: Matching and price competition. American Economic Review 96, 652--668 (2006)
\bibitem{chen} Chen, X., Deng, X.: Settling the complexity of two-player Nash equilibrium. In: FOCS 2006 (2006)
\bibitem{chen1} Chen, X., Deng, X., Teng, S.-H.: Computing Nash equilibria: Approximation and smoothed complexity. In: FOCS
2006 (2006) 
\bibitem{dan} Dantzig, G.B.: Linear programming and extensions. Princeton University Press (1963)
\bibitem{gov} Govindan, S., Wilson, R.: A global Newton method to compute Nash equilibria. Journal of Economic Theory 110(1),
65--86 (2003)
\bibitem{g2} Govindan, S., Wilson, R.: Equivalence and invariance of the index and degree of Nash equilibria. Games and Economic Behavior
21(1), 56--61 (1997) 
\bibitem{sur} Herings, P.J., Peeters, R.: Homotopy Methods to Compute Equilibria in Game Theory. Econ Theory 42, 119--156
(2010)
\bibitem{kan} Kannan, R., Theobald, T.: Games of fixed rank: a hierarchy of bimatrix games. In: SODA 2007, 1124--1132 (2007)
\bibitem{koh} Kohlberg, E., Mertens, J.-F.: On the strategic stability of equilibria. Econometrica 54(5), 1003--1037 (1986)
\bibitem{ks} Kontogiannis, S., Spirakis, P.: Exploiting concavity in bimatrix Games: new polynomially tractable subclasses.
In: APPROX 2010 (2010)
\bibitem{lem} Lemke, C.E., Howson, J.T.: Equilibrium points of bimatrix games. Journal of the Society for Industrial and
Applied Mathematics 12, 413--423 (1964)
\bibitem{lip} Lipton, R.J., Markakis, E., Mehta, A.: Playing large games using simple strategies. In: EC 2003, 36--41 (2003) 
\bibitem{nash} Nash, J.: Non-cooperative games. Annals of Mathematics 54, 289--295 (1951)
\bibitem{agt} Nisan, N., Roughgarden, T., Tardos, E., Vazirani, V.V. (Eds.): Algorithmic Game Theory. Cambridge University
Press (2007)
\bibitem{pap} Papadimitriou, C.H.: Algorithms, games and the Internet. In: STOC 2001 (2001)
\bibitem{pap1} Papadimitriou, C.H.: On the complexity of the parity argument and other inefficient proofs of existence.
Journal of Computer \& System Sciences 48, 498--532 (1992)
\bibitem{ak} Predtetchinski, A.: A general structure theorem for the Nash equilibrium correspondence. Games and Economic
Behavior 66(2), 950--958 (2009)
\bibitem{sav} Savani, R., Stengel, B.V.: Hard-to-solve bimatrix games. Econometrica 74(2), 397--429 (2006)
\bibitem{stengel} Schemde, A.V., Stengel, B.V.: Strategic Characterization of the Index of an Equilibrium. In: SAGT 2008,
242--254 (2008) 
\bibitem{shap} Shapley, L.S.: A note on the Lemke-Howson algorithm. Mathematical Programming Study 1: Pivoting and
Extensions, 175--189 (1974)
\bibitem{the} Theobald, T: Enumerating the Nash equilibria of rank-1 games. Polyhedral Computation, CRM Proceedings, American
Mathematical Society (2007)
\end{thebibliography}
\end{document}